\newif\ifANON
\ANONtrue
\ANONfalse  

\documentclass[11pt]{article}
\usepackage[margin=1in]{geometry}
\usepackage{graphicx}
\graphicspath{ {../} }
\usepackage{amsmath,amssymb,amsthm}
\usepackage{thm-restate}
\usepackage{hyperref}
\usepackage{todonotes}

\hypersetup{colorlinks=true,allcolors=blue}
\usepackage{cleveref}
\usepackage[normalem]{ulem}

\theoremstyle{plain}
\newtheorem{theorem}{Theorem}[section]
\newtheorem{lemma}[theorem]{Lemma}
\newtheorem{claim}[theorem]{Claim}
\newtheorem{corollary}[theorem]{Corollary}
\newtheorem{proposition}[theorem]{Proposition}

\theoremstyle{definition}
\newtheorem{definition}[theorem]{Definition}

\theoremstyle{remark}
\newtheorem{remark}[theorem]{Remark}

\newcommand{\remove}[1]{}

\newcommand{\R}{\mathbb{R}}
\newcommand{\CALP}{\mathcal{P}}
\newcommand{\calP}{\mathcal{P}}
\newcommand{\calQ}{\mathcal{Q}}
\newcommand{\calR}{\mathcal{R}}
\newcommand{\norm}[1]{\lVert#1\rVert}
\newcommand{\bigpar}[1]{\bigl(#1\bigr)}
\newcommand{\bigsqpar}[1]{\bigl[#1\bigr]}
\DeclareMathOperator{\ddim}{ddim}
\DeclareMathOperator{\diam}{diam}
\DeclareMathOperator{\dist}{dist}
\DeclareMathOperator{\poly}{poly}
\DeclareMathOperator{\cost}{cost}
\DeclareMathOperator{\E}{\mathbb{E}}
\newcommand{\OPT}{\mathrm{OPT}}

\newcommand{\sumbad}{\sum_{(x,y)\in S}}

\usepackage{xcolor}

\title{Dimension Reduction and Coresets for Min-Sum \texorpdfstring{$k$}{k}-Clustering}

\ifANON
\author{Anonymous Authors}
\date{}
\else
\author{Asaf Dallal
    \thanks{Email: \texttt{asaf.dallal@weizmann.ac.il}
    }
    \qquad
    Robert Krauthgamer
        \thanks{The Harry Weinrebe Professorial Chair of Computer Science.
            Work partially supported by the Israel Science Foundation grant \#1336/23.
            Email: \texttt{robert.krauthgamer@weizmann.ac.il}
        }
    \\  Weizmann Institute of Science
}
\fi

\begin{document}
\maketitle

\begin{abstract}
The min-sum $k$-clustering problem asks to partition an input of $n$ points 
into $k$ clusters so as to minimize the total pairwise distance within clusters.
It is a natural alternative to center-based objectives when the data 
is not well captured by a small set of representative points. 
We study for this problem two standard sketching methods, 
dimension reduction and coresets,
which represent the input more compactly while preserving 
the min-sum cost of every $k$-partition up to $1\pm\varepsilon$ factor. 

Our first contribution addresses dimension reduction in Euclidean space. 
The Johnson-Lindenstrauss (JL) Lemma readily implies that 
a random linear map to target dimension $t=O(\varepsilon^{-2}\log n)$ preserves, 
with high probability, the min-sum cost of all $k$-partitions. 
We show that target dimension
$t=O(\varepsilon^{-2}(\ddim\cdot\log(1/\varepsilon)+\log k+\log\log n))$
suffices for Euclidean inputs of doubling dimension at most $\ddim$, 
and moreover this bound is nearly optimal. 
We also show that target dimension $t=O(\varepsilon^{-2}(\log(1/\varepsilon)+\log k+\log\Delta))$
suffices for Euclidean inputs of aspect ratio at most $\Delta$. 

Our second contribution is to design the first $\varepsilon$-coreset
for metric min-sum $k$-clustering. 
It is actually a movement-based coreset,
has size $(1/\varepsilon)^{O(\ddim)} k\log n$ 
for inputs of doubling dimension at most $\ddim$, 
and can be constructed in randomized near-linear-time.
We further use these coresets to design streaming algorithms for the problem.
\end{abstract}

\section{Introduction}

Clustering is a fundamental task in geometry and data analysis, and it is commonly formulated as partitioning a dataset $X$ of points into $k$ groups,
called clusters, so that points in the same group are similar to one another.
We study min-sum $k$-clustering, where the cost of each cluster is computed 
directly from the pairwise distances between its points,
in contrast to center-based clustering, like $k$-median and $k$-means,
where the cost of each cluster is computed using a center point.  This min-sum objective can model better clusters that are coherent but not well represented by centers; in particular, it does not impose a
Voronoi (nearest-neighbor) structure on the partition. 
For instance, min-sum 
$k$-clustering has been used to recover a grouping of data from noisy pairwise similarity information, 
and to cluster large collections of protein sequences~\cite{Voevodski2011MinSumProtein, Voevodski2012ActiveClustering}.

Formally, in the min-sum \(k\)-clustering problem (also called $k$-min-sum), 
the input is an $n$-point metric space \((X,\dist)\)
and the goal is to find a $k$-partition \(\calP=\{P_1,\ldots,P_k\}\) of \(X\), 
that minimizes the objective  
\begin{equation} \label{eq:minsum}
        \cost(\calP)
    =
    \sum_{i=1}^k
    \sum_{x,y\in P_i}
    \dist(x,y).
\end{equation}
To simplify notation, we sum over all ordered pairs, which includes every pair $\{x,y\}\subseteq P_i$ twice.
The input $X$ can actually be a subset of points from a larger metric space, 
such as $\R^d$ equipped with a distance function that is derived from an $\ell_p$ norm or from a kernel function. 
A well-known example where the min-sum objective is highly effective consists of points spread evenly over two co-centric rings that form two natural clusters,
even though each ring has no natural center point that is close to its own points and far from the other ring, as illustrated in \Cref{fig:ring-example}.
This effectiveness has led to increased interest in min-sum $k$-clustering,

including algorithms in various models and metric spaces, 
and even in (non-metric) squared-Euclidean distance~\cite{Schulman00, BCR01, CzumajSohler2010, banerjee2020min, NaderiRS24, AmanSP24, lee2024approximability}.

Nevertheless, the known algorithmic tools for this objective
  are still lagging behind those for center-based objectives,
  which may explain why the min-sum objective is used much less in practice.
  
Notice that the total distance between all pairs of points in $X$ is split into inter-cluster pairs and intra-cluster pairs, 
thus minimizing the objective~\eqref{eq:minsum} is equivalent to maximizing its complement, 
which is the (metric) Max-$k$-Cut problem and whose special case \(k=2\) is the famous Max-Cut problem.
Note, however, that this equivalence does not extend to multiplicative approximation of the optimum.

\begin{figure}[t]
    \centering
    \includegraphics[width=0.82\linewidth]{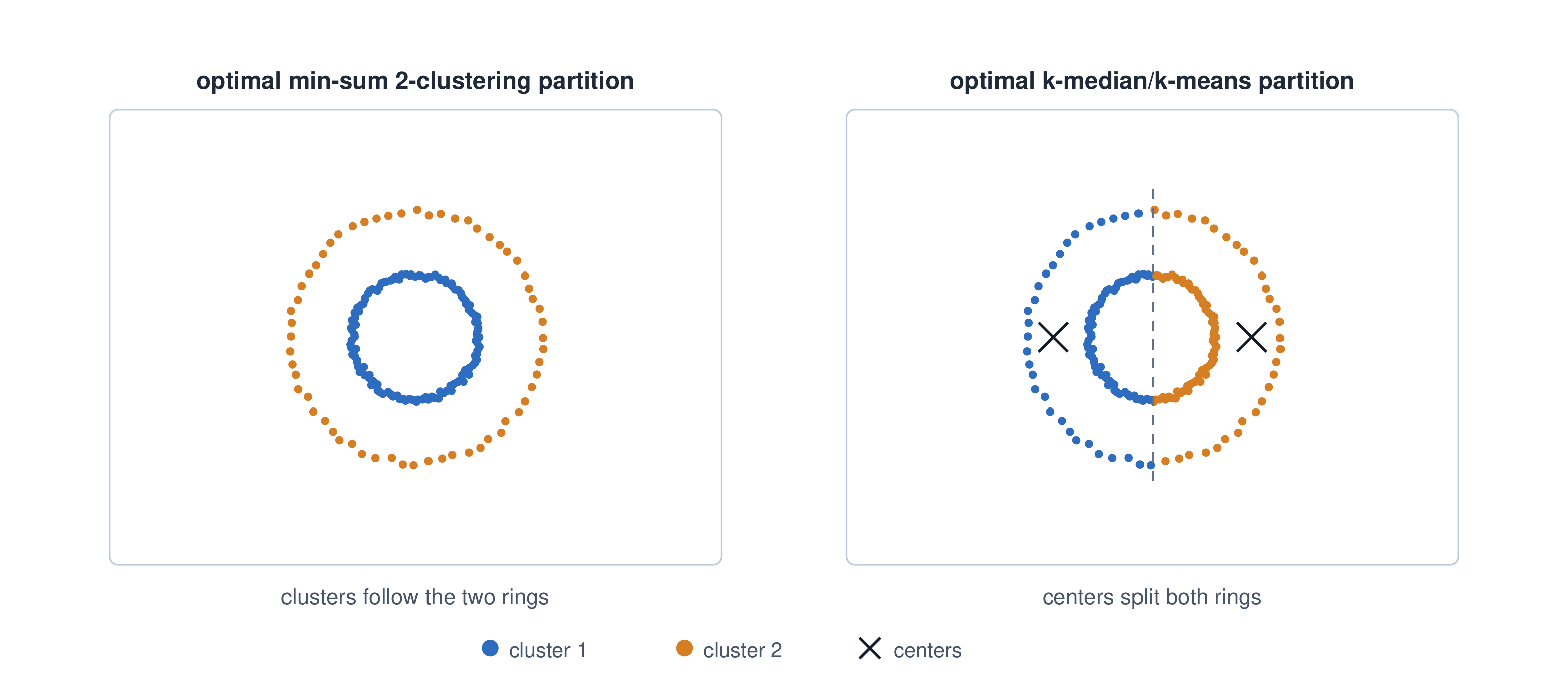}
    \caption{A min-sum $k$-partition can differ significantly from a \(k\)-median or \(k\)-means partition.  Here \(k=2\) in Euclidean $\R^2$; the outer ring has
    radius \(\alpha>1\) times that of the inner ring and contains \(1/\alpha\)
    as many points. In the min-sum optimum, each ring is assigned its own cluster, whereas the center-based optimum is induced by two centers and cuts
    both rings.
    }
    \label{fig:ring-example}
\end{figure}

A sketch is a small or simple representation of the input that approximately preserves some quantities of interest.
Sketching is a fundamental technique in algorithm design,
and is especially important for processing large datasets, 
where it can significantly improve the running time, storage, and/or communication costs,
which strongly depend on the size of the input representation.

Many different sketching methods were developed,
and they have found numerous applications,
from geometric optimization and numerical linear algebra to streaming
and distributed computation~\cite{Muthukrishnan05,Woodruff14,Phillips17}.
We study sketches for min-sum \(k\)-clustering and seek a for-all guarantee: 
\emph{the sketch preserves the cost of all \(k\)-partitions simultaneously.} 
This would follow easily from the much stronger requirement that all pairwise distances are preserved (multiplicatively). 
We study two common sketching approaches: 
dimension reduction for inputs in Euclidean $\R^d$, which reduces the ambient dimension \(d\), and 
coresets for metric inputs, which reduce the dataset size \(n\). 
Both of them are defined via a map from the input metric to a more "compact" metric, 
in which the sketch is evaluated. 
Via this map, every $k$-partition of $X$ induces a $k$-partition of its image and vice versa, as follows. 
Let \(g:X\to Y\) be a map to a metric space $(Y, \rho)$. For a set
\(A\subseteq X\), let \(g(A)\) be the multiset $\{g(x) \mid x\in A\}$ 
and for a $k$-partition \(\calP=\{P_1,\ldots,P_k\}\) of \(X\), let
\[
    g(\calP)=\{g(P_1),\ldots,g(P_k)\}.
\]
We view this as a $k$-partition of $g(X)$, even though the multisets $g(P_i)$ and $g(P_j)$ for $i\neq j$ might not be disjoint.

A useful way to measure the intrinsic geometry of a metric space $X$ is its \emph{doubling
dimension}, denoted by $\ddim(X)$,
defined as the smallest \(D\geq1\) such that every metric ball 
can be covered by at most \(2^D\) balls of half the radius \cite{GKL03}.
This notion is effective even for inputs in Euclidean $\R^d$, 
as their doubling dimension is at most $O(d)$ but can sometimes be significantly smaller than $d$.
Another important parameter is the \emph{aspect ratio} of a set $S\subseteq X$, defined as
\[
\Delta=\frac{\diam(S)}{\min_{x\neq y\in S}\dist(x,y)},
\]
where \(\diam(S)=\max_{x,y\in S}\dist(x,y)\).
Many sketching results in the literature depend on the doubling dimension or aspect ratio of the input $X$, 
and are thus particularly effective when they are low.

\paragraph{Dimension reduction.}

A standard sketching method is oblivious dimension reduction:
When the input is a high-dimensional Euclidean point set, 
it is useful to project it to a low-dimensional space 
before running downstream applications like a clustering algorithm.
This is a very powerful algorithmic tool,
because the efficiency of many geometric algorithms
depends heavily on the ambient dimension. 
The Johnson-Lindenstrauss (JL) Lemma~\cite{johnson1984extensions} states that 
every \(n\)-point Euclidean metric can be mapped (embedded) into \(\mathbb{R}^t\), 
for \(t=O(\varepsilon^{-2}\log n)\), 
while preserving all pairwise distances within factor \(1\pm\varepsilon\). 
It is known that the above bound on the dimension $t$ is tight in general 
for preserving all pairwise distances~\cite{larsen2016optimality}. 

A standard oblivious construction, which we use below, 
is to apply a random linear map \(f:\mathbb{R}^d\to\mathbb{R}^t\), 
given by a matrix whose entries are independent Gaussian random variables with appropriate scaling, 
and this is known to succeed with high probability. 

The JL guarantee on all pairwise distances clearly implies 
that also the min-sum cost of every partition is preserved up to the same factor. 
However, preserving all the pairwise distances might be an excessive
requirement, and indeed dimension reduction tailored to clustering and
related geometric optimization problems has been studied
extensively.
In particular, for center-based clustering such as $k$-means and $k$-median, 
it suffices to use dimension $t$ that depends on $k$ instead of $n$, 
namely, $t=O(\varepsilon^{-2}\log(k/\varepsilon))$~\cite{Boutsidis2010,CohenEMMP15,DBLP:conf/stoc/BecchettiBC0S19, MakarychevMR19}.

The picture changes however when the center set is restricted: 
for $k$-median with discrete centers (medoids), the target dimension also depends 
on the doubling dimension and on a \(\log\log n\) term \cite{JKSSY26}. 
For min-sum \(k\)-clustering, no analogous dimension-reduction result 
was previously known, and it was open to obtain any guarantee beyond the JL bound.
The main obstacle is the indirect dependence of the objective on cluster sizes, 
as each point \(x\) in cluster \(C\) contributes \(O(|C|)\) terms to the objective. 
In particular, placing \(x\) in a smaller cluster (that has the same average distance from $x$) 
will decrease the contribution of \(x\) to the min-sum objective. 
This is very different from center-based objectives, 
where each point contributes only one term to the objective, its distance to its assigned center. 
Thus the known center-based arguments do not apply, at least not directly.

\paragraph{Coresets.}

Another standard sketching method is a coreset,
where the input $X$ is replaced by a smaller (but weighted) set. 
We propose the following definition, which guarantees correspondence between $k$-partitions of the input $X$ and these of the coreset. Notice that it applies to a metric input (not necessarily Euclidean). 

\begin{definition}[Movement-Based Coreset]\label{def:movement-based-coresets}
Let \((X,\dist)\) be a finite metric space.  
A weighted set \((S,w)\), with \(S\subseteq X\) 
is a \emph{movement-based $\varepsilon$-coreset} 
for min-sum \(k\)-clustering on \(X\) 
if there exists a map \(\pi:X\to S\) that is consistent with the weights, 
namely, $w(s)=|\pi^{-1}(s)|$ for all $s\in S$, 
such that
\[
    \forall \text{$k$-partition $\calP$ of $X$},
    \qquad
    \cost(\pi(\calP))
    \in
    (1\pm\varepsilon)\cost(\calP).
\]
\end{definition}

We call \(\pi\) a \emph{movement rule}; it need not be stored explicitly as part
of the coreset, i.e., the size of the coreset is $|S|$. 
It is easy to see that a $k$-partition of $S$ that splits the weight of a point \(s\in S\) among multiple parts can be replaced, without increasing the min-sum cost, 
by a $k$-partition that does not split the weight of $s$ 
(i.e., its entire weight is assigned to a single part); see \Cref{clm:unsplitting-representatives}. 
Consequently, it suffices to solve the clustering problem on the small coreset $(S,w)$ 
and pull back its solution through the movement rule, 
to a solution for the original input $X$. 
The drawback is that the coreset must be interpreted together with a movement rule, however,
in our constructions this rule is simply the nearest-neighbor (in $S$), and hence $\pi^{-1}(s)$ can be computed from $X$ and $S$, and need not be stored explicitly. 

\Cref{def:movement-based-coresets} is analogous to the usual strong coreset guarantee for center-based
clustering~\cite{Har-PeledM04,frahling2005coresets,BravermanFLSY17}, 
where a weighted set $S$ represents the input $X$ 
if it preserves the cost of every possible choice of centers.  
The key difference is that a solution for min-sum \(k\)-clustering
is not a set of centers but rather a partition of \(X\),
and we thus need a movement rule that specifies a correspondence 
between partitions of \(X\) and partitions of the coreset.
A downside of this movement-based definition is that it does not support 
sampling, reweighting, or similar modifications of the dataset. 
The same approach of a movement rule was used in previous coreset definitions,
particularly for Max-Cut in~\cite{frahling2005coresets}
and for balanced $k$-median in~\cite{BateniBLM14} (called a "mapping coreset"),
although their accuracy requirements are different and not suitable for our min-sum objective.

Previously, Czumaj and Sohler~\cite{CzumajSohler2010} proposed 
a very different definition of a coreset for min-sum $k$-clustering. 
While it is also based on a movement rule from the input $X$ to the coreset $S$,
their coreset $S$ is equipped with modified distances 
and thus need not be a submetric of $X$
(for example, $S$ might not be Euclidean even if $X$ is).
Moreover, they do not require that the min-sum cost of every partition is preserved, 
and in this sense their definition is weaker than ours.
Altogether, the two definitions are incomparable.

\subsection{Our Results}
\label{sec:our-results}

Our first two results are dimension-reduction bounds for min-sum \(k\)-clustering,
one in terms of the doubling dimension of the input, 
and one in terms of its aspect ratio. 
No bounds were previously known for this problem. 
Throughout, a Gaussian JL map \(f:\R^d\to\R^t\) refers to a random linear map 
given by a matrix whose entries are independent \(N(0,1/t)\) random variables.

\begin{restatable}[Euclidean Dimension Reduction for Doubling Sets]{theorem}{MainJLDDim}
\label{thm:main-jl-ddim}
Let $X\subseteq\R^d$ be an instance of min-sum $k$-clustering in Euclidean space, and let
$0<\varepsilon<1/4$.  
Let $f:\R^d\to\R^t$ be a Gaussian JL map, where
$
    t\ge
    c\varepsilon^{-2} \bigpar{\ddim(X)\log(1/\varepsilon)+\log k+\log\log n}
$,
for a sufficiently large constant $c>0$.  
Then, with probability at least $4/5$,
\[
    \forall \text{$k$-partition $\calP$ of $X$}, 
    \qquad 
    \cost\bigpar{f(\CALP)}
    \in 
    (1\pm\varepsilon)\cost(\CALP).
\] 
\end{restatable}

It is interesting to identify other parameters that can be used to guarantee dimension reduction. 
Surprisingly, we find that the aspect ratio can serve as such a parameter, 
in contrast to some other clustering objectives 
where standard lower-bound examples have aspect ratio $O(1)$. 
The aspect ratio is a natural complement to the doubling dimension, 
since it limits the number of distance scales that the min-sum objective must handle, 
rather than the packing complexity at every single scale.

\begin{restatable}[Euclidean Dimension Reduction for Bounded Aspect Ratio]{theorem}{MainJLAspect}
\label{thm:main-jl-aspect-ratio}
Let $X\subseteq\R^d$ be an instance of min-sum $k$-clustering in Euclidean space with aspect ratio at most $\Delta$, and let $0<\varepsilon<1/4$.  
Let $f:\R^d\to\R^t$ be a Gaussian JL map, where
$
    t\ge
    c \varepsilon^{-2} \bigpar{\log(1/\varepsilon)+\log k+\log \Delta}
$
for a sufficiently large constant $c>0$. 
Then, with probability at least $4/5$,
\[
    \forall \text{$k$-partition $\calP$ of $X$}, 
    \qquad \cost\bigpar{f(\CALP)}
    \in
    (1\pm\varepsilon)\cost(\CALP).
\]
\end{restatable}

We remark that \Cref{thm:main-jl-ddim,thm:main-jl-aspect-ratio} 
immediately extend to objectives where pairwise distances are raised to power $p\ge1$, 
i.e., \(\|x-y\|_2^p\), 
with target dimension \(t=O(p^2\varepsilon^{-2}(\ddim(X)\log(p/\varepsilon)+\log k+\log\log n+p))\) in the doubling-dimension bound, 
and \(t=O(p^4\varepsilon^{-2}(\log(p/\varepsilon)+\log k+p\log(\Delta)))\) in the aspect-ratio bound;
see \Cref{sec:l2p-generalization}. 
An immediate corollary is that Schulman's ($1+\varepsilon$)-approximation algorithms for \(\ell_2^2\)-costs~\cite{Schulman00}, 
whose running time is exponential in the ambient dimension \(d\), 
can use the reduced dimension $t$ instead of \(d\).

We further show that our bound on the target dimension in \Cref{thm:main-jl-ddim} is near-optimal (for Gaussian JL maps),
namely, \(\Omega\bigpar{\ddim(X)}\) is necessary even for min-sum \(2\)-clustering, 
and our dependence on $n$, namely \(\Omega(\varepsilon^{-2}\log\log n)\) 
is needed even for constant doubling dimension; 
see \Cref{thm:lower-bounds}.

Our next result is a movement-based coreset for min-sum \(k\)-clustering 
in doubling metrics (not necessarily Euclidean),
that has size $(1/\varepsilon)^{O(\ddim)}k\log n$. 
We further provide a near-linear-time implementation,
which uses internally yet another movement-based coreset that we design. 
These two coresets are constructed using different techniques,
our main result is based on a density-adaptive method, 
and the internal coreset uses a simpler annular construction.

In comparison, the coreset of Czumaj and Sohler~\cite{CzumajSohler2010}
achieves $O(1)$-approximation (whereas ours achieves $1+\varepsilon$),
its size is $\poly(k\log n)$ with no dependence on the doubling dimension,
and its construction time is also near-linear in \(n\).

But as mentioned earlier, their definition of a coreset is very different 
and incomparable to ours. 

\begin{restatable}[Density-adaptive coreset]{theorem}{DensityAdaptiveCoreset}
\label{thm:density-adaptive}
Let $(X,\dist)$ be an $n$-point metric space of doubling dimension at most $\ddim\geq 1$. 
Then for every $0<\varepsilon<\frac{1}{2}$, 
there is a movement-based \(\varepsilon\)-coreset \((S,w)\) 
of size
\[
    |S|\le \left(\frac{1}{\varepsilon}\right)^{O(\ddim)} k\log n,
\]
whose movement rule is the nearest-neighbor map.
Moreover, such a coreset can be constructed in randomized
near-linear time with high probability.
\end{restatable}
 
We further show that the logarithmic dependence on \(n\) in these bounds is necessary, 
as even for min-sum \(2\)-clustering in $\R$ (dimension one), 
every constant-factor movement-based coreset must have size \(\Omega(\log n)\); 
see \Cref{prop:movement-logn-necessary}. 
To bypass this lower bound, 
one may follow the spirit of \Cref{thm:main-jl-aspect-ratio}
and construct coresets according to our definition, 
whose size depends on other parameters, like the aspect ratio $\Delta$. 

Finally, we use our coreset results to derive two streaming algorithms,
one for insertion-only streams and one for dynamic streams. 
Our algorithms estimate the optimum value of min-sum \(k\)-clustering,
denoted $\OPT(X)$, 
but since they essentially maintain a coreset of the input, 
they can also output a partition of the coreset.
In fact, given further access to the original input $X$ (as queries or another pass), 
they can even pull the clustering back to $X$, 
by applying the nearest-neighbor movement rule. 

In the insertion-only streaming model, 
the input $X$ is presented as a sequence of points in arbitrary order,
and the algorithm has oracle access to the distance function between points. 

Unfortunately, we cannot simply apply merge-and-reduce, which is a standard technique
to devise a streaming algorithm from any coreset that has two properties called Merge and Reduce
(see \cite{Har-PeledM04, Phillips17}).

It is easy to see that movement-based coresets satisfy the following Reduce Property 
(summarized as ``a coreset of a coreset is itself a coreset''): 
suppose $S$ is a coreset for an input $X$, 
and $T$ is a coreset for the coreset $S$;
then $T$ is a coreset for the original input $X$, 
with the errors accumulating.  

But unfortunately, movement-based coresets do not necessarily satisfy the following Merge Property (summarized as ``the union of coresets is a coreset''): 
if each $S_i$ is a coreset of $X_i$, say for $i=1,2$,
then $\cup_i S_i$ is a coreset for $\cup_i X_i$.

Nevertheless, we can combine the two-pass streaming algorithm of Czumaj and Sohler~\cite{CzumajSohler2010} 
with our annular construction of a coreset 
to obtain a three-pass insertion-only streaming algorithm, as follows. 

\begin{restatable}[Algorithm for insertion-only streams]{corollary}{ThreePassAnnularStreaming}
\label{cor:three-pass-annular-streaming}
There is a randomized three-pass algorithm for insertion-only streams that,
given as input an $n$-point metric space $(X,\dist)$ of doubling dimension at most $\ddim\geq 1$ and $0<\varepsilon<\frac{1}{10}$, 
outputs with high probability an estimate 
$\widehat{\OPT} \in (1\pm\varepsilon)\ \OPT(X)$
using $(k\log n)^{O(1)} + \left(\frac{k}{\varepsilon}\right)^{O(\ddim)}\log n$ space.
\footnote{We do not state the time complexity, which may indeed be exponential,
as this may require an offline algorithm that solves min-sum $k$-clustering 
on the final coreset within approximation $1\pm\varepsilon$ (or optimally). 
} 
Moreover, with one additional pass, the algorithm can output a $k$-partition of $X$, 
whose cost is within $1+\varepsilon$ factor of the optimum. 
\end{restatable}

In the dynamic streaming model, the input $X$ is presented as 
a sequence of updates, i.e., point insertions and deletions. 
Our result for dynamic streams applies to inputs lying in a predetermined Euclidean space, such as a grid $[\Delta]^d$ for $\Delta=\poly(n)$. 
For this setting, we use our dimension reduction (\Cref{thm:main-jl-ddim}) together with annular coreset construction 
and obtain an \(O(\log n)\)-pass dynamic streaming algorithm, as follows.

\begin{restatable}[Algorithm for dynamic Euclidean streams]{corollary}{MultiPassDynamicStreaming}
\label{cor:multipass-dynamic-streaming}
There is a randomized multipass algorithm for dynamic streams that,
given as input an \(n\)-point Euclidean set $X\subseteq[\Delta]^d$  with
\(\Delta=\poly(n)\) of doubling dimension at most \(\ddim\geq 1\), and \(0<\varepsilon<\frac{1}{10}\),
outputs with probability at least $3/4$ an estimate 
$\widehat{\OPT} \in (1\pm\varepsilon)\ \OPT(X)$,
using \(O(\log n)\) passes and \(  \left(kt/\varepsilon\right)^t\log n
  + \operatorname{poly}\bigl(d,k,t,\log\log(n\Delta)\bigr)\) space
for \(t=O(\varepsilon^{-2}(\ddim\log(1/\varepsilon)+\log k+\log\log n))\). 

Moreover, with one additional pass, the algorithm can output a $k$-partition of $X$, 
whose cost is within $1+\varepsilon$ factor of the optimum. 
\end{restatable}

\subsection{Related Work}
\label{sec:related-work}

The min-sum \(k\)-clustering problem has been studied from algorithmic perspectives.
The best polynomial-time approximation known for the problem is $O(\log n)$\cite{BehsazFSS19},
and there is also an $O(1)$ bicriteria approximation algorithm~\cite{BCR01}. 
Improved approximations are known for specific families of metric spaces, 
such as doubling metrics and planar graphs \cite{NaderiRS24},
for small $k$, namely $k=o(\log n/\log\log n)$ \cite{CzumajSohler2010},
and also for combinations thereof~\cite{Schulman00, lee2024approximability}. 

In~\cite{CzumajSohler2010} there are also other types of algorithms, 
including sublinear-time approximation algorithm
(when the input $n$-point metric is represented explicitly and has size $n^2$), 
an algorithm for constructing a coreset, 
and a two-pass algorithm for insertion-only streams.

The Johnson--Lindenstrauss lemma~\cite{johnson1984extensions} has motivated a line of research on dimension reduction for Euclidean clustering problems. 
For some problems, preserving all pairwise distances among the input points
immediately implies that the clustering objective value is preserved
(e.g., min-sum $k$-clustering),
whereas for other problems, there is no immediate connection 
(e.g., $k$-median, because the pairwise guarantee does not include the centers). 
In both cases, one may hope to show a similar or even better dimension reduction result,
and in many known dimension-reduction results the target dimension is indeed smaller,
and often depends on parameters of the problem (e.g., $k$) 
or the input (e.g., doubling dimension) rather than on \(n\). 
The studied clustering problems include 
$k$-means and $k$-median \cite{Boutsidis2010, CohenEMMP15, DBLP:conf/stoc/BecchettiBC0S19, MakarychevMR19},
discrete $k$-median ($k$-medoids) \cite{JKSSY26},
$k$-center \cite{JKS24},
facility location and minimum-weight spanning tree \cite{NarayananSIZ21, HuangJKY25},
and more \cite{lammersen2010approximation, izzo2021dimensionality, DBLP:journals/ijcga/HuangQYD22, CW25, GJKSSSW25}.

Dimension reduction with additive-error guarantees has also been
studied for Euclidean sets of bounded doubling
dimension~\cite{IN07,AHPY13}. For applications of dimension reduction beyond clustering, see~\cite{Dirksen16}.

Coresets and related streaming algorithms clustering problems have been studied extensively.
Har-Peled and Mazumdar~\cite{Har-PeledM04} gave one of the early coreset 
constructions for Euclidean $k$-median and $k$-means,
and further used this coreset to obtain insertion-only streaming algorithms. 
Their coreset is based on moving the input points to nearby grid points,
where the grid granularity depends on an approximate solution,
and is actually a movement-based coreset.
Bateni, Bhaskara, Lattanzi, and Mirrokni~\cite{BateniBLM14} gave 
a coreset construction for balanced $k$-center, $k$-median, and $k$-means,
and used it to develop distributed algorithms. 
Their coreset too is movement-based, although they call it a mapping coreset. 
Importantly, the error of a coreset is defined in these works 
in a manner that is tailored for the $k$-median objective, 
but not useful for min-sum $k$-clustering. 

Yet another method for constructing coresets, that is clearly not movement-based, 
relies on sampling points \cite{Chen09, FeldmanL11},
and extensive subsequent work has improved and refined these methods,
producing fantastic coresets for these clustering problems and many other problems.
While this literature is too vast to list explicitly, 
we point out that it includes streaming algorithms 
for high-dimensional $k$-median and $k$-means \cite{BravermanFLSY17,hu2018nearly}. 

\subsection{Technical Overview}
\label{sec:technical-overview}

We first describe the dimension-reduction argument. 
Since a Gaussian JL map $f$ need not preserve all the pairwise distances,
we associate with \(f\) a graph on the $n$ input points,
whose edges are exactly the pairs whose distances are distorted by more than a
\(1\pm\varepsilon\) factor. 
The sparse-subgraph theorem of Makarychev,
Makarychev, and Razenshteyn \cite{MakarychevMR19} implies that 
inside any set of points (think of it as a possible cluster), 
one can find a small set of vertices to remove, 
such that the remaining subgraph has small maximum degree. 
This sparsity is enough for min-sum \(k\)-clustering because 
the distortion of any other edge inside such a graph can be bounded,
using the triangle inequality, through (the average of) many good intermediate points in the same cluster.

The main difficulty is to bound pairs that involve removed points,
as these are not controlled by the sparse-graph argument. 
The key idea is to fix a reference partition, 
and use it to replace the contribution of such pairs 
with quantities that are stable under the map $f$. 

In the doubling-dimension theorem, 
this is implemented by fixing an optimal partition 
and considering the \(1\)-median centers of its \(k\) clusters as base points: 
we build nets at several scales around these centers, 
and for each pair involving at least one removed point, we orient the pair toward the endpoint whose reference cluster has larger cardinality. 
The endpoints of the pair are replaced by nearby net representatives, 
at scales determined by this oriented reference cluster, 
and the error from this replacement is charged either to the pair itself or to the cost of the reference partition.  Our key new idea is to use an optimal clustering, and in particular its
$1$-median centers, as a reference for bounding the approximation error,
even though the min-sum objective itself does not involve centers.

For the aspect-ratio theorem, we use the same principle but now 
the reference partition is one whose clusters have roughly the same size,
and the bounded aspect ratio is used (in the analysis) 
to find such a partition whose cost is not much bigger than the optimal cost.
Each pair with a removed endpoint is now replaced by the average distance 
between the two corresponding reference clusters. 
Along the way, we use a known property of a Gaussian JL map $f$, 
that the sum of multiple pairwise distances (i.e., a fixed collection of pairs)
is with high probability preserved by $f$
and moreover has the same concentration bound as a single pair of points.

The two lower bounds are based on a common principle:
If a Gaussian JL map $f$ shortens even one important direction (pair of points), 
the min-sum objective can magnify the error by clustering one endpoint with many copies of the other.
It follows that applying the map $f$ produces an instance whose cost 
is much smaller than the original cost.

Now, for the \(\Omega(\ddim)\) lower bound, 
we place many orthogonal directions around a single point of high multiplicity. 
For the \(\Omega(\log\log n)\) lower bound with constant doubling dimension, 
we place the same obstruction at exponentially separated scales, 
with increasing multiplicities across the scales.

We proceed to describe our two movement-based coresets. 
The annular construction runs in near-linear time,
however its size is $(k/\varepsilon)^{O(\ddim)}\log n$,
which has worse dependence on \(k\),
and we therefore use it a fast preprocessing step
to speed up our other construction.

This annular construction adapts the discretization approach that
Har-Peled and Mazumdar~\cite{Har-PeledM04} devised for a $k$-median coreset.

It first finds an $O(1)$-approximate $k$-partition, to be used as a reference,
and then builds nets in geometrically decreasing balls
around the ``centers'' of that reference partition. 
The movement rule maps every point to its nearest net point. 
We again analyze an arbitrary partition by orienting each pair toward the endpoint whose reference cluster is larger.
The movement of the other endpoint can then be charged either to the pair itself or to the reference cluster cost.
This argument bounds the movement-error, as required for a coreset,
however it loses a factor \(k\),
and thus an \(\varepsilon\)-movement guarantee requires
running the construction with accuracy parameter \(\varepsilon/k\) rather than \(\varepsilon\), 
leading to the larger \((k/\varepsilon)^{O(\ddim)}\log n\) size bound. 
A fast implementation requires computing the reference solution efficiently,
and we show how to bootstrap the required reference solution by
alternating coarse annular coresets with weighted bicriteria balanced
$k$-median computations.

Our main coreset contribution is the density-adaptive construction
that avoids the term \(k^{O(\ddim)}\)
and achieves the size bound of \Cref{thm:density-adaptive}.
It does not require an approximate solution as reference,
but rather only an estimate $\widehat{\OPT}$ for the optimal value,
which is easier to obtain.
Its key idea is to choose how far to move each point based on its local density,
as follows. 

For each point \(x\in X\), define
\[
    \rho(x)
    =
    \min_{2\le m\le n}
    \max\left\{
        \frac1m\sum_{v\in \Gamma_m(x)}\dist(x,v),
        \frac{\widehat{\OPT}}{km^2}
    \right\},
\]
where \(\Gamma_m(x)\) is the set of the \(m\) nearest points to \(x\). The value
\(\rho(x)\) is a density-adaptive movement scale. The nearest-neighbor term
captures the amount of movement that can be charged to the local distances
around \(x\), while the additive term \(\widehat{\OPT}/(km^2)\) is the largest
movement scale that can be absorbed directly for clusters of size about \(m\):
moving \(m\) points by this amount affects \(m^2\) pairwise terms and costs only
\(O(\widehat{\OPT}/k)\).

Our coreset algorithm group points $x\in X$ according to the magnitude of \(\rho(x)\), and build a separate net for
each group (at the corresponding scale). Points with small \(\rho(x)\) are
represented more accurately, while points with larger \(\rho(x)\) can be moved
farther without violating the movement bound.
\subsection{Conclusions}

Our dimension-reduction and coreset results preserve substantially
more than the optimal clustering value. In particular, our
movement-based coreset maps the points to representatives such that the costs of all $k$-partitions are preserved simultaneously.
Our constructions in doubling metrics enforce an even stronger
condition: every point moves only a small distance relative to
an appropriate local scale. We use packing bounds in doubling
metrics to control the number of representatives needed to
satisfy these movement bounds which results in exponential dependence on the doubling dimension in
our coreset size.

This motivates a weaker notion focused on preserving the
optimum. We call a weighted subset $(S,w)$ a \emph{weak $\varepsilon$-coreset} for min-sum \(k\)-clustering of $X$, if
\[  
    \operatorname{OPT}(S,w)
  \in (1\pm\varepsilon)\operatorname{OPT}(X),
\]
and there exists a polynomial-time algorithm that converts every
$\gamma$-approximate solution on $(S,w)$ into a
$\gamma(1+\varepsilon)$-approximate solution on $X$.

We don't know whether this weaker notion admits
smaller coresets, particularly with improved
dependence on $n$ or the doubling dimension. Since min-sum \(k\)-clustering and balanced \(k\)-median are equivalent up to factor two,
coreset techniques for balanced \(k\)-median may help address this question.

\section{Preliminaries}
\label{sec:preliminaries}

\subsection{Nets and Doubling Metrics}
\label{sec:preliminaries-nets}

Let $(X,\dist)$ be a metric space.
For $x\in X$ and $r>0$, we write
$
    B_X(x,r)=\{y\in X:\dist(x,y)\le r\}
$,
and omit the subscript when the ambient space $X$ is clear.

\begin{definition}[$r$-net]
\label{def:nets-packings}
Let $(X,\dist)$ be a metric space and let $r>0$.  A subset $N\subseteq X$ is
called an \emph{$r$-packing} if
\[
    \forall u\neq v\in N,\qquad \dist(u,v)>r .
\]
A subset $N\subseteq X$ is called an \emph{$r$-covering} of $X$ if
\[
    \forall x\in X,\exists u\in N, \text{ such that } x\in B_X(u,r).
\]
A subset $N\subseteq X$ is called an \emph{$r$-net} of $X$ if it is both an
$r$-packing and an $r$-covering of $X$.
\end{definition}

\begin{lemma}[Packing Property {\cite[Proposition~1.1]{GKL03}}]
\label{lem:doubling-packing-bound}
Let $P$ be a metric space and let $N\subseteq P$ be an $r$-packing.  Then
\[
    |N|\le \left(\frac{\diam(P)}{r}\right)^{O(\ddim(P))}.
\]
\end{lemma}

\subsection{Gaussian JL Maps}
\label{sec:preliminaries-gaussian-jl}

\begin{definition}[Random Dimension Reduction]
\label{def:random-dim-red}
A random map $f:\R^d\to\R^t$ is called an
\emph{$(\varepsilon,\delta)$-random dimension reduction} for $X$ if
\[
  \forall A,B\subseteq X,
  \qquad
    \Pr\left[
        \sum_{x\in A, y\in B}\norm{f(x)-f(y)}
        \in
        (1\pm\varepsilon)
        \sum_{x\in A, y\in B}\norm{x-y}
    \right]
    \ge
    1-\delta .
\]
It is called an \emph{$(\varepsilon,\delta,\rho)$-random dimension reduction} for $X$
if it additionally satisfies
\[
    \forall x,y \in X,\qquad
    \E\Big[
        \max\left\{0, 
            \norm{f(x)-f(y)}-(1+\varepsilon)\norm{x-y}
        \right\}
    \Big]
    \le
    \rho\norm{x-y}.
\]
\end{definition}

\begin{definition}[Local Ball-Tail Property]
\label{def:local-ball-tail}
Let $f:\R^d\to\R^t$ be a random map.  We say that $f$ has the
\emph{$\mu$-local ball-tail property} for $X$ if
\[
    \forall y\in X,\ r>0,\qquad
    \Pr\Big[
        \exists x\in X\cap B(y,r),\ 
        \norm{f(x)-f(y)}\ge 10r
    \Big]
    \le
    e^{-\mu}.
\]
\end{definition}

\begin{lemma}[Gaussian Pair Tail {\cite[Theorem~3.1]{Matousek08}}]
\label{lem:gaussian-pair-tail}
There is a universal constant $c>0$ such that, for every $0<\varepsilon,\delta<1$ the following holds.  Let $f:\R^d\to\R^t$ be a Gaussian JL map with
$t\ge c\varepsilon^{-2}\log(1/\delta)$.  Then for every fixed $x,y\in\R^d$,
\[
    \Pr\!\left[
        \norm{f(x)-f(y)}
        \notin
        (1\pm\varepsilon)\norm{x-y}
    \right]
    \le \delta .
\]
\end{lemma}

\begin{lemma}[Gaussian Set-Pair Preservation {\cite[Subsection~A.1]{Danos21}}]
\label{lem:gaussian-set-pair-preservation}
There is a universal constant $c>0$ such
that, for every $0<\varepsilon,\delta<1$ the following holds.  Let $X\subseteq\R^d$ be finite, and let
$f:\R^d\to\R^t$ be a Gaussian JL map with
$
t\ge c\varepsilon^{-2}\bigpar{\log(1/\varepsilon)+\log(1/\delta)}
$.
Then $f$ is an $(\varepsilon,\delta)$-random dimension reduction for $X$.
\end{lemma}

\begin{lemma}[Gaussian One-Sided Excess {\cite[Lemmas~C.1 and C.2]{MakarychevMR19}}]
\label{lem:gaussian-one-sided-excess}
There is a universal constant $c>0$ such
that, for every $0<\varepsilon,\rho<1$ the following holds.  Let $f:\R^d\to\R^t$ be a Gaussian JL map with
$
t\ge c\varepsilon^{-2}\bigpar{\log(1/\varepsilon)+\log(1/\rho)}
$.
Then
\[
    \forall x,y \in \R^d,\qquad
    \E\Big[
        \max\left\{0, 
            \norm{f(x)-f(y)}-(1+\varepsilon)\norm{x-y}
        \right\}
    \Big]
    \le
    \rho\norm{x-y}.
\]
\end{lemma}

\begin{lemma}[Gaussian Local Ball-Tail in Doubling Sets {\cite{IN07}}]
\label{lem:gaussian-local-ball-tail}
There is a universal constant $c>0$ such that
if $X\subseteq\R^d$ has doubling dimension at most $\ddim\ge1$,
and $f:\R^d\to\R^t$ is a Gaussian JL map with $t\ge c(\ddim+\mu)$,
then $f$ has the $\mu$-local ball-tail property for $X$.
\end{lemma}

\section{Dimension Reduction for Min-Sum \texorpdfstring{$k$}{k}-Clustering}
\label{sec:dimension-reduction}

In this section we prove two dimension-reduction theorems, stated below for convenience. The first depends on the doubling
dimension of the data, and the second depends on the aspect ratio.

\MainJLDDim*

\MainJLAspect*

\subsection{Preliminaries: Distortion Graph and Related Tools}
\label{sec:dr-random-subset-template}

Both dimension-reduction upper bounds rely on an intermediate guarantee, which shows that even after the map \(f\) is revealed, every random subset \(C\subseteq X\) has its min-sum cost preserved up to a controlled additive term. The tools below implement this strategy: the distortion graph records the pairs distorted by \(f\), and the sparsity of this graph lets us control their total contribution to the min-sum cost.
\begin{definition}[Distortion graph]
\label{def:distortion-graph}
Given $0<\varepsilon\leq 1/2$, the \emph{distortion graph} $G=(X,E)$ for a map $f$ has vertex set $X$ and an
edge $(x,y)\in E$ whenever
\[
    \norm{f(x)-f(y)}
    \notin
    (1\pm\varepsilon)\norm{x-y}.
\]
Observe that when $f$ is an $(\varepsilon,\delta)$-random dimension reduction for $X$,
\[
    \forall x,y \in X,\qquad
    \Pr\bigl[(x,y)\in E\bigr]\le \delta .
\]
\end{definition}

\begin{definition}[Everywhere-sparse graph]
A graph $H=(V,E)$ is called \emph{$\theta$-everywhere sparse} if
\[
    \forall u \in V, \quad\deg_H(u)\le \theta |V|.    
\]
Notice that this definition allows $V$ to be empty.
\end{definition}

\begin{theorem}[Sparse Subgraph Theorem {\cite[Theorem~3.2]{MakarychevMR19}}]
\label{thm:sparse-graph}
Consider a finite set $X$ and a random graph $H=(V,E)$, where $V$ is a random
subset of $X$ and $E$ is a random set of edges between vertices in $V$.  No
independence assumptions are made.  Let $\theta\in(0,1/2)$, and assume that
\[
    \forall x,y \in X,\qquad
    \Pr\bigsqpar{(x,y)\in E}\le\delta,
\]
where $\delta\le \theta^7/600$.  Then there exists a random subset
$V'\subseteq V$, defined on the same probability space, such that
\begin{itemize}
    \item $H[V']$ is $\theta$-everywhere sparse, and
    \item $\Pr\bigsqpar{u\in V\setminus V'}\le\theta$ for all $u\in X$.
\end{itemize}
\end{theorem}

\begin{lemma}[Sparse Graph Preservation]
\label{lem:sparse-graph-preservation}
Let $\theta\le 1/10$,  $0<\varepsilon\le 1/2$ and let $G=(Y,E)$ be the distortion graph of $Y$ for a map $f$. If $G$ is $\theta$-everywhere sparse, then
\begin{align}
    \sum_{x,y\in Y}\norm{f(x)-f(y)}
    \le
    (1+5\theta)(1+2\varepsilon)
    \sum_{x,y\in Y}\norm{x-y},\label{eq:first-sparse-graph-preservation}
\end{align}
and
\begin{align}
    \sum_{x,y\in Y}\norm{x-y}
    \le
    (1+5\theta)(1+2\varepsilon)
    \sum_{x,y\in Y}\norm{f(x)-f(y)}.\label{eq:second-sparse-graph-preservation}
\end{align}
\end{lemma}

\begin{proof}
For each $x\in Y$, let $I_x=\{y\in Y: (x,y)\notin E\}.$
Since $G$ is $\theta$-everywhere sparse, we have
\[
    \forall x,y \in Y,\qquad
    |I_x\cap I_y|\ge (1-2\theta)|Y|.
\]
Now fix $x,y\in X$. For every $z\in I_x\cap I_y$,
\[
    \norm{f(x)-f(y)}
    \le
    \norm{f(x)-f(z)}+\norm{f(z)-f(y)}
    \le
    (1+\varepsilon)\bigpar{\norm{x-z}+\norm{z-y}},
\]
 and averaging this over all $z\in I_x\cap I_y$ gives
\[
    \norm{f(x)-f(y)}
    \le
    \frac{1+\varepsilon}{(1-2\theta)|Y|}
    \sum_{z\in I_x\cap I_y}
    \bigpar{\norm{x-z}+\norm{z-y}}.
\]
Now split the sum $\underset{x,y\in X}{\sum}\norm{f(x)-f(y)}$ into undistorted and distorted pairs.  The undistorted
pairs contribute at most $(1+\varepsilon)\sum_{x,y\in Y}\norm{x-y}$.  For the
distorted pairs, use the above averaging bound to charge each undistorted pair at most $4\theta |Y|$ times. Hence
\begin{gather}
    \sum_{\substack{x,y\in Y, y\notin I_x}}\norm{f(x)-f(y)}
    \le
    (1+\varepsilon)\frac{4\theta}{1-2\theta}
    \sum_{\substack{x,y\in Y, y\in I_x}}\norm{x-y}\leq 5\theta(1+\varepsilon)\sum_{x,y\in Y}\norm{x-y}, \label{eq:distorted_cost_bound}
\end{gather}
where the second inequality follows from $\theta\leq1/10$. Inequality \eqref{eq:distorted_cost_bound} together with the bound on the distorted pairs contribution proves \eqref{eq:first-sparse-graph-preservation}. Inequality \eqref{eq:second-sparse-graph-preservation} follows in the same way, using
\[
    z\in I_x\cap I_y,\quad \norm{x-y}
    \le
    \frac{1}{1-\varepsilon}
    \bigpar{\norm{f(x)-f(z)}+\norm{f(z)-f(y)}},
\]
and the conclusion follows from
$1/(1-\varepsilon)\le 1+2\varepsilon$.
\end{proof}

\begin{lemma}[One-Median Charge]
\label{lem:k-median-centers-cost-upper-bound}
Let $S\subseteq\R^d$, and let $c^*$ be an optimal center for $S$ under the
$1$-median objective.  Then
\[
    |S|\sum_{y\in S}\norm{y-c^*}
    \le
    \sum_{y,y'\in S}\norm{y-y'}.
\]
\end{lemma}

\begin{proof}
Since $c^*$ is an optimal $1$-median center for $S$,
\[
    \forall y^\prime\in S,\qquad
    \sum_{y\in S}\norm{y-c^*}
    \le
    \sum_{y\in S}\norm{y-y'}.
\]
Summing over all $y^\prime\in S$ gives the lemma.
\end{proof}

\subsection{Upper Bound for Doubling Dimension}

The next lemma is the core of the doubling-dimension argument. It shows that
for every random subset \(C\subseteq X\), even one chosen after seeing the
embedding, the min-sum costs of \(C\) and \(f(C)\) are close up to a small
additive term controlled by the optimal clustering. We first use \Cref{thm:sparse-graph} and \Cref{lem:sparse-graph-preservation}
to obtain a subset $C'\subseteq C$ whose cost is approximately
preserved.
For pairs involving points outside $C'$, we replace their endpoints
by net points around the optimal cluster centers and orient each
pair toward the endpoint whose optimal cluster is larger, allowing
us to bound the resulting error in terms of $\cost(C)$ and $\OPT$.

\begin{lemma}[Preservation of a Random Cluster]
\label{lem:ddim-random-subset-preservation}
Let $X\subseteq\R^d$ be an instance of min-sum $k$-clustering of doubling dimension at most
$\ddim\geq1$, and let $\CALP^*=\{P_1^*,\ldots,P_k^*\}$ be an optimal partition of $X$ with
$\OPT=\cost(\CALP^*)$.  Assume $0<\varepsilon\le 1/2$ and
\[
    \theta\le \frac1{20},
    \qquad
    \delta
    \le
    \min\left\{
        \frac{\theta^7}{600},
        \frac{\varepsilon^{\alpha\ddim}}{100k^2\log^2 n}
    \right\},
    \qquad
    \mu
    \ge
    \beta\bigpar{\ddim\log(1/\varepsilon)+\log k+\log\log n}
\]
for sufficiently large constants $\alpha,\beta>0$.  Let
$f:\R^d\to\R^t$ be a random map such that
\[
    \forall x,y \in X,\qquad
    \Pr\!\left[
        \norm{f(x)-f(y)}
        \notin
        (1\pm\varepsilon)\norm{x-y}
    \right]\le\delta,
\]
and assume that $f$ has the $\mu$-local ball-tail property for $X$.  Then for
every random subset $C\subseteq X$ which may depend on $f$, with probability at
least $4/5$,
\begin{align}
    \cost(C)
    &\le
    A\bigpar{
        \cost(f(C))
        +
        O\bigl(\varepsilon(k\theta+n^{-8})\OPT\bigr)
    },
    \label{eq:ddim-subset-first}
    \\
    \cost(f(C))
    &\le
    A\bigpar{
        \cost(C)
        +
        O\bigl(\varepsilon(k\theta+n^{-8})\OPT\bigr)
    },
    \label{eq:ddim-subset-second}
\end{align}
where $A=(1+5\theta)(1+2\varepsilon)$.
\end{lemma}

\begin{proof}
Let $c_j^*$ be a $1$-median center of $P_j^*$, and set
$r_0=20\OPT$.  Pick $m_0=O(\log n)$ such that
$2^{m_0}=40n^{10}$, and for $i=0,\ldots,m_0$ set $r_i=r_0/2^i$.
Let $N_{ij}$ be an $(\varepsilon r_i/100)$-net of
$X\cap B(c_j^*,r_i)$, and denote
\[
    N=\bigcup_{i=0}^{m_0}\bigcup_{j=1}^k N_{ij}.
\]
By \Cref{lem:doubling-packing-bound}, applied to each $N_{ij}$,
\[
    |N|\le k\varepsilon^{-O(\ddim)}\log n .
\]
Using the pairwise distortion assumption and taking
a union bound over $N\times N$, we obtain
\begin{equation}
\label{eq:ddim-net-preservation}
    \Pr\bigl[
        \exists u,v\in N,\,\,\,
        \norm{f(u)-f(v)}
        \notin
        (1\pm\varepsilon)\norm{u-v}
    \bigr]
    \le
    \delta |N|^2
    \le
    \frac1{100}.
\end{equation}
Additionally, since $f$ has the $\mu$-local ball-tail property, for every $i,j$
and every $y\in N_{ij}$,
\[
    \Pr\!\left[
        \exists x\in X\cap B(y,\varepsilon r_i/100),\,\,
        \norm{f(x)-f(y)}>\varepsilon r_i/10
    \right]
    \le
    e^{-\mu}.
\]
For a sufficiently large constant $\beta$, another union bound gives
\begin{equation}
\label{eq:ddim-net-ball}
    \Pr\!\left[
        \exists i,j,y\in N_{ij},\
        \exists x\in X\cap B(y,\varepsilon r_i/100),\,\,\,
        \norm{f(x)-f(y)}>\varepsilon r_i/10
    \right]
    \le
    \frac1{100}.
\end{equation}
Let $\Lambda$ be the event that neither \eqref{eq:ddim-net-preservation} nor
\eqref{eq:ddim-net-ball} occurs.  Then $\Pr(\Lambda)\ge 9/10$.
Let $C$ be a random subset of $X$, and let $G[C]$ be the distortion graph of
$C$.  By \Cref{thm:sparse-graph}, there exists $C'\subseteq C$ such that
$G[C']$ is $\theta$-everywhere sparse and
\[
    \forall u \in X, \qquad \Pr \bigsqpar{u\in C\setminus C'}\le \theta.
\]

Consider a pair $x,y\in C$.  Assume without loss of generality that
$y\in P_\ell^*$, $x\in P_j^*$, and $|P_\ell^*|\ge |P_j^*|$.  We define
$r(x,y)>0$ and $u(x,y)\in N$ as follows.  If $\norm{x-y}\le 2\OPT$, let
$r(x,y)$ be the smallest radius $r_i$ such that
\[
    \norm{x-c_\ell^*}\le r_i,
\]
and let $u(x,y)$ be the closest point to $x$ in the corresponding net
$N_{i\ell}$.  Otherwise, let $r(x,y)=r_0$ and let $u(x,y)$ be the nearest point
to $x$ in the coarsest net $N_{0j}$.  
In the far case \(\norm{x-y}>2\OPT\), while
\(\norm{y-c_\ell^*}\le \OPT\), and hence
\(\norm{x-c_\ell^*}>\OPT\).
The quantities $r(y,x)$ and $u(y,x)$ are defined only from $y$ and the same
center $c_\ell^*$.  Namely, $r(y,x)$ is the smallest radius level such that
$\norm{y-c_\ell^*}\le r_i,$ and $u(y,x)$ is the closest point to $y$ in the corresponding net.
By construction,
\begin{equation}
\label{eq:ddim-radius-bound}
    r(x,y)\le
    \max\left\{25\norm{x-c_\ell^*},\frac{\OPT}{2n^{10}}\right\},
    \qquad
    r(y,x)\le
    \max\left\{2\norm{y-c_\ell^*},\frac{\OPT}{2n^{10}}\right\}.
\end{equation}
Thus
\begin{align}
    r(x,y)+r(y,x)
    &\le
    \frac{\OPT}{n^{10}}
    +25\norm{x-c_\ell^*}
    +2\norm{y-c_\ell^*}
    \notag\\
    &\le
    \frac{\OPT}{n^{10}}
    +25\norm{x-y}
    +30\norm{y-c_\ell^*}.
    \label{eq:ddim-radius-sum}
\end{align}
Therefore
\begin{align}
    \left|
    \norm{x-y}-\norm{u(x,y)-u(y,x)}
    \right|
    &\le
    \norm{x-u(x,y)}+\norm{y-u(y,x)}
    \notag\\
    &\le
    \frac{\varepsilon}{100}\bigl(r(x,y)+r(y,x)\bigr)
    \notag\\
    &\le
    \varepsilon\left(
        \frac{\OPT}{n^{10}}
        +\frac{\norm{x-y}}{4}
        +\frac{3\norm{y-c_\ell^*}}{10}
    \right).
    \label{eq:ddim-rounding-original}
\end{align}
If event $\Lambda$ occurs, the failure event in
\eqref{eq:ddim-net-ball} does not occur. Hence
$\norm{f(x)-f(u(x,y))}\le \varepsilon r(x,y)/10$ and
$\norm{f(y)-f(u(y,x))}\le \varepsilon r(y,x)/10$.
The triangle inequality and \eqref{eq:ddim-radius-sum} therefore give:
\begin{align}
    \left|
    \norm{f(x)-f(y)}
    -
    \norm{f(u(x,y))-f(u(y,x))}
    \right|
    &\le
    \norm{f(x)-f(u(x,y))}
    +
    \norm{f(y)-f(u(y,x))}
    \notag\\
    &\le
    \frac{\varepsilon}{10}\bigl(r(x,y)+r(y,x)\bigr)
    \notag\\
    &\le
    \frac{\varepsilon}{10}
    \left(
        \frac{\OPT}{n^{10}}
        +25\norm{x-y}
        +30\norm{y-c_\ell^*}
    \right).
    \label{eq:ddim-rounding-image}
\end{align}

Define
\[
    g(x,y)=
    \begin{cases}
        x, & x,y\in C',\\
        u(x,y), & \text{otherwise}.
    \end{cases}
\]
Then set
\[
    \cost_g(C)=\sum_{x,y\in C}\norm{g(x,y)-g(y,x)}
\]
and
\[
    \cost_g(f(C))
    =
    \sum_{x,y\in C}
    \norm{f(g(x,y))-f(g(y,x))}.
\]
Let
$S=\{(x,y)\in C\times C:\ x\notin C'\text{ or }y\notin C'\}.$
If event $\Lambda$ occurs, using \Cref{lem:sparse-graph-preservation} on $C'$ and the net
preservation event on the rounded pairs,
\begin{align}
    \cost_g(f(C))
    &\le
    (1+5\theta)(1+2\varepsilon)
    \sum_{x,y\in C'}\norm{x-y}
    +(1+\varepsilon)
    \sumbad\norm{u(x,y)-u(y,x)}
    \notag\\
    &\le
    A\cost_g(C),
    \label{eq:ddim-cost-g-forward}
\end{align}
and similarly
\begin{equation}
\label{eq:ddim-cost-g-backward}
    \cost_g(C)
    \le
    A\cost_g(f(C)).
\end{equation}

By \eqref{eq:ddim-rounding-original} and \eqref{eq:ddim-rounding-image},
\begin{align}
    \cost(C)
    &\le
    \cost_g(C)
    +
    \varepsilon\frac{\OPT}{n^8}
    +
    \frac{\varepsilon}{4}\cost(C)
    +
    \frac{3\varepsilon}{10}
    \sumbad\norm{y-c^*(y)},
    \label{eq:ddim-compare-one}
    \\
    \cost_g(C)
    &\le
    \cost(C)
    +
    \varepsilon\frac{\OPT}{n^8}
    +
    \frac{\varepsilon}{4}\cost(C)
    +
    \frac{3\varepsilon}{10}
    \sumbad\norm{y-c^*(y)},
    \label{eq:ddim-compare-two}
    \\
    \cost_g(f(C))
    &\le
    \cost(f(C))+10\varepsilon\frac{\OPT}{n^8}+3\varepsilon\cost(C)+3\varepsilon
    \sumbad\norm{y-c^*(y)},
    \label{eq:ddim-compare-three}
    \\
    \cost(f(C))
    &\le
    \cost_g(f(C))
    +
    10\varepsilon\frac{\OPT}{n^8}+3\varepsilon\cost(C)+3\varepsilon
    \sumbad\norm{y-c^*(y)}.
    \label{eq:ddim-compare-four}
\end{align}
Here $c^*(y)$ denotes the $1$-median center of the optimal cluster containing
$y$, and the term $\OPT/n^8$ comes from summing $\OPT/n^{10}$ over at most
$n^2$ ordered pairs.

It remains to bound $\sumbad\norm{y-c^*(y)}$ in
\eqref{eq:ddim-compare-one}--\eqref{eq:ddim-compare-four}.
By the second conclusion of \Cref{thm:sparse-graph},
$\Pr\bigsqpar{u\in C\setminus C'}\le\theta$ for every $u\in X$.
The definition of $S$ and a union bound therefore give
\[
    \forall x,y \in X,\qquad
    \Pr\bigsqpar{(x,y)\in S}\le 2\theta.
\]
Orient the pair so that $y$ belongs to an optimal cluster at least as large as
the optimal cluster of $x$.  Then
\begin{align}
    \E\left[
        \sumbad\norm{y-c^*(y)}
    \right]
    &\le
    2\theta
    \sum_{\substack{x,y\in X\\ |P_x^*|\le |P_y^*|}}
    \norm{y-c^*(y)}
    \notag\\
    &\le
    2\theta
    \sum_{y\in X} k|P_y^*|\norm{y-c^*(y)}
    \notag\\
    &\le
    2\theta k
    \sum_{i=1}^k\sum_{u,v\in P_i^*}\norm{u-v}
    =
    2\theta k\OPT,
    \label{eq:ddim-charge-expectation}
\end{align}
where the last inequality is \Cref{lem:k-median-centers-cost-upper-bound}.
By Markov's inequality, with probability at least $9/10$,
\[
    \sumbad\norm{y-c^*(y)}
    \le
    20\theta k\OPT.
\]
Together with $\Lambda$, this event holds with probability at least $4/5$.
Combining \eqref{eq:ddim-cost-g-forward}--\eqref{eq:ddim-compare-four} and
rescaling $\varepsilon$ by a sufficiently small universal constant gives
\eqref{eq:ddim-subset-first} and \eqref{eq:ddim-subset-second}. This concludes the proof of \Cref{lem:ddim-random-subset-preservation}.
\end{proof}

\begin{lemma}[From random clusters to all partitions]
\label{lem:bad-cluster-extraction}
Suppose that for every random subset $C\subseteq X$ which may depend on $f$, with
probability at least $4/5$,
\[
    \cost(C)\le A\bigpar{\cost(f(C))+\Gamma}
    \qquad\text{and}\qquad
    \cost(f(C))\le A\bigpar{\cost(C)+\Gamma},
\]
where $A\le (1+\varepsilon)^2$ and
$A\Gamma\le \frac{\varepsilon}{2k}\OPT .$
Then, with probability at least $4/5$, for every partition $\CALP$ of $X$,
\begin{gather*}
    (1-\varepsilon)\cost(\CALP)
    \le
    (1+\varepsilon)^2\cost(f(\CALP)),
    \\
    \cost(f(\CALP))
    \le
    (1+\varepsilon)^3\cost(\CALP).
\end{gather*}
\end{lemma}

\begin{proof}
Denote the event in the statement by $D$.  We define a random cluster
$C\subseteq X$ as follows.  If $D$ happens, set $C=\emptyset$.  Otherwise, there
is a partition $\CALP=\{P_1,\ldots,P_k\}$ such that at least one of the
following holds:
\begin{gather*}
    \cost(\CALP)
    >
    (1+\varepsilon)^2\cost(f(\CALP))
    +
    \varepsilon\cost(\CALP)
    \ge
    A\cost(f(\CALP))+\varepsilon\cost(\CALP),
    \\
    \cost(f(\CALP))
    >
    (1+\varepsilon)^3\cost(\CALP)
    =
    (1+\varepsilon)^2\cost(\CALP)
    +
    \varepsilon(1+\varepsilon)^2\cost(\CALP)
    \ge
    A\cost(\CALP)+\varepsilon\cost(\CALP).
\end{gather*}
Thus there is some $i\in[k]$ such that one of the following holds:
\begin{align}
    \cost(P_i)
    &>
    A\cost(f(P_i))
    +
    \frac{\varepsilon}{k}\cost(\CALP),
    \label{eq:bad-cluster-original}
    \\
    \cost(f(P_i))
    &>
    A\cost(P_i)
    +
    \frac{\varepsilon}{k}\cost(\CALP).
    \label{eq:bad-cluster-image}
\end{align}
Set $C=P_i$ for such an index $i$.

If the random-subset preservation inequalities hold for this choice of $C$, then
\[
    \cost(P_i)
    \le
    A\cost(f(P_i))+A\Gamma
    \le
    A\cost(f(P_i))+\frac{\varepsilon}{2k}\cost(\CALP),
\]
and
\[
    \cost(f(P_i))
    \le
    A\cost(P_i)+A\Gamma
    \le
    A\cost(P_i)+\frac{\varepsilon}{2k}\cost(\CALP),
\]
because $\cost(\CALP)\ge\OPT$.  This contradicts
\eqref{eq:bad-cluster-original} and \eqref{eq:bad-cluster-image}.  Therefore, if event
$D$ fails, then the random-subset preservation lemma must fail for the random
cluster $C$.  Since this happens with probability at most $1/5$, the event $D$
holds with probability at least $4/5$, proving the lemma.
\end{proof}

\begin{proof}[Proof of \Cref{thm:main-jl-ddim}]
Let $\varepsilon'=(1+\varepsilon)^{1/3}-1=O(\varepsilon),$  and choose
\[
    \theta
    =
    \min\left\{
        \frac{(\varepsilon')^2}{20},
        \frac{1}{2(1+\varepsilon')^2k^2}
    \right\}.
\]
Let
\[
    \delta
    \le
    \min\left\{
        \frac{\theta^7}{600},
        \frac{(\varepsilon')^{\alpha\ddim}}{100k^2\log^2 n}
    \right\}.
\]
By \Cref{lem:gaussian-pair-tail}, for a sufficiently large constant $c>0$, the
Gaussian JL map in the theorem satisfies
\[
    \forall x,y \in X,\qquad
    \Pr\!\left[
        \norm{f(x)-f(y)}
        \notin
        (1\pm\varepsilon')\norm{x-y}
    \right]\le\delta.
\]
By \Cref{lem:gaussian-local-ball-tail}, the same map has the $\mu$-local
ball-tail property with
\[
    \mu
    \ge
    \beta\bigpar{\ddim\log(1/\varepsilon')+\log k+\log\log n}.
\]
We can thus apply \Cref{lem:ddim-random-subset-preservation}.
Given our choice of $\theta$ above and after rescaling constants,
we may choose $\Gamma=\varepsilon'\OPT/(2k(1+\varepsilon')^2)$.
Since $A\le(1+\varepsilon')^2$, this choice satisfies
$A\Gamma\le\varepsilon'\OPT/(2k)$, as required by
\Cref{lem:bad-cluster-extraction}.
Applying that lemma with $\varepsilon'$ proves the theorem.
\end{proof}

\subsection{Upper Bound for Aspect Ratio}

The proof has three steps. First, in
\Cref{lem:aspect-random-subset-preservation}, we assume the existence of a
$\lambda$-approximate reference partition in which every cluster contains
at least \(n/(2k)\) points, and use the sparse-distortion graph and
average distances within reference clusters to prove approximate
preservation of the cost of a random cluster, up to an additive error.
Second, in \Cref{lem:balanced-partition-approximation}, we use the bounded
aspect ratio to supply such a partition with $\lambda=4k^2\Delta$.
Finally, we use \Cref{lem:bad-cluster-extraction} again to obtain
simultaneous preservation of the cost of every $k$-partition.
It remains to set $\theta$ so that the additive error in
\Cref{lem:aspect-random-subset-preservation} satisfies the condition of
\Cref{lem:bad-cluster-extraction}, yielding the logarithmic dependence
on $\Delta$.

\begin{lemma}[Preservation of a Random Cluster]
\label{lem:aspect-random-subset-preservation}
Let $X\subseteq\R^d$ be an instance of min-sum $k$-clustering.  Let
$f:\R^d\to\R^t$ be an $(\varepsilon,\delta,\rho)$-random dimension reduction for
$X$, where $0<\varepsilon\le 1/2$ and
\[
    \theta\le \frac1{20},
    \qquad
    \delta\le
    \min\left\{\frac{\theta^7}{600},\frac1{10k^3}\right\},
    \qquad
    \rho\le\theta.
\]
Assume that there is a $\lambda$-approximation partition
$\CALP^*=\{P_1^*,\ldots,P_k^*\}$ such that $\min_{i\in[k]} |P_i^*|\ge \frac{n}{2k}.$
Then for every random subset $C\subseteq X$ which may depend on $f$, with
probability at least $4/5$,
\begin{align}
    \cost(C)
    &\le
    A\bigpar{
        \cost(f(C))
        +
        400(1+\varepsilon)\lambda\theta k\,\OPT
    },
    \label{eq:aspect-subset-first}
    \\
    \cost(f(C))
    &\le
    A\bigpar{
        \cost(C)
        +
        400(1+\varepsilon)\lambda\theta k\,\OPT
    },
    \label{eq:aspect-subset-second}
\end{align}
where $A=(1+5\theta)(1+2\varepsilon)$.
\end{lemma}

\begin{proof}
Let $\OPT$ denote the optimal cost.  Since $\CALP^*$ is a $\lambda$-approximation,
$\cost(\CALP^*)\le \lambda\cdot\OPT$. Let $\Lambda$ be the event that for all $i,j\in[k]$,
\[
    (1-\varepsilon)
    \sum_{x\in P_i^*,\,y\in P_j^*}\norm{x-y}
    \le
    \sum_{x\in P_i^*,\,y\in P_j^*}\norm{f(x)-f(y)}
    \le
    (1+\varepsilon)
    \sum_{x\in P_i^*,\,y\in P_j^*}\norm{x-y}.
\]
Since $f$ is an $(\varepsilon,\delta)$-random dimension reduction,
\[
    \Pr(\Lambda)\ge 1-k^2\delta\ge 9/10.
\]

Let $C$ be a random subset of $X$, and let $G[C]$ be the distortion graph of
$C$.  By \Cref{thm:sparse-graph}, there exists $C'\subseteq C$ such that
$G[C']$ is $\theta$-everywhere sparse and
\[
    \forall u\in X, \qquad \Pr\bigsqpar{u\in C\setminus C'}\le \theta.
\]
Let
\[
    S=\{(x,y)\in C\times C:\ x\notin C'\text{ or }y\notin C'\}.
\]
For $x\in X$, let $P_x^*$ be the cluster of $\CALP^*$ containing $x$.  Define
\[
    g(x,y)=
    \begin{cases}
        \norm{x-y}, & x,y\in C',\\[1ex]
        \displaystyle
        \frac{1}{|P_x^*||P_y^*|}
        \sum_{x'\in P_x^*,\,y'\in P_y^*}\norm{x'-y'},
        & \text{otherwise},
    \end{cases}
\]
and
\[
    g_f(x,y)=
    \begin{cases}
        \norm{f(x)-f(y)}, & x,y\in C',\\[1ex]
        \displaystyle
        \frac{1}{|P_x^*||P_y^*|}
        \sum_{x'\in P_x^*,\,y'\in P_y^*}\norm{f(x')-f(y')},
        & \text{otherwise}.
    \end{cases}
\]
Let
\[
    \cost_g(C)=\sum_{x,y\in C}g(x,y),
    \qquad
    \cost_g(f(C))=\sum_{x,y\in C}g_f(x,y).
\]
Using \Cref{lem:sparse-graph-preservation} for $C'$ and
conditioning on the event $\Lambda$, we obtain
\begin{align}
    \cost_g(f(C))
    &\le
    (1+5\theta)(1+2\varepsilon)
    \sum_{x,y\in C'}\norm{x-y}
    \notag\\
    &\quad
    +(1+\varepsilon)
    \sumbad
    \frac{1}{|P_x^*||P_y^*|}
    \sum_{x'\in P_x^*,\,y'\in P_y^*}\norm{x'-y'}
    \notag\\
    &\le
    A\cost_g(C),
    \label{eq:aspect-g-forward}
\end{align}
and similarly
\begin{equation}
\label{eq:aspect-g-backward}
    \cost_g(C)
    \le
    A\cost_g(f(C)).
\end{equation}

The triangle inequality gives
\begin{align}
    \forall x,y \in X,\qquad
    &\left|
    \norm{x-y}
    -
    \frac{1}{|P_x^*||P_y^*|}
    \sum_{x'\in P_x^*,\,y'\in P_y^*}\norm{x'-y'}
    \right|
    \notag\\
    &\qquad\le
    \frac1{|P_x^*|}
    \sum_{x'\in P_x^*}\norm{x-x'}
    +
    \frac1{|P_y^*|}
    \sum_{y'\in P_y^*}\norm{y-y'}.
    \label{eq:aspect-original-rounding}
\end{align}
The same inequality, with $f(x)$ and $f(y)$ in place of $x$ and $y$, bounds $\left|\norm{f(x)-f(y)}-g_f(x,y)\right|$ similarly.  Combining these inequalities gives
\[
    \cost(C)\le A\bigpar{\cost(f(C))+\mathcal D},
    \qquad
    \cost(f(C))\le A\bigpar{\cost(C)+\mathcal D},
\]
where
\begin{align*}
    \mathcal D
    =
    \sumbad
    \Bigg(
    &\frac1{|P_x^*|}\sum_{x'\in P_x^*}\norm{x-x'}
    +
    \frac1{|P_y^*|}\sum_{y'\in P_y^*}\norm{y-y'}
    \\
    &+
    \frac1{|P_x^*|}\sum_{x'\in P_x^*}\norm{f(x)-f(x')}
    +
    \frac1{|P_y^*|}\sum_{y'\in P_y^*}\norm{f(y)-f(y')}
    \Bigg).
\end{align*}
It remains to show that $\mathcal D$ cannot be too large.  We have
\[
    \forall x,y \in X,\qquad
    \Pr\bigsqpar{(x,y)\in S}\le 2\theta.
\]
Hence
\begin{align}
    \E\left[
        \sumbad
        \frac1{|P_x^*|}
        \sum_{x'\in P_x^*}\norm{x-x'}
    \right]
    &\le
    2\theta\sum_{x,y\in X}
        \frac1{|P_x^*|}
        \sum_{x'\in P_x^*}\norm{x-x'}
    \notag\\
    &=
    2\theta\sum_{x\in X}
        \frac{n}{|P_x^*|}
        \sum_{x'\in P_x^*}\norm{x-x'}
    \notag\\
    &\le
    4\theta k
    \sum_{x\in X}\sum_{x'\in P_x^*}\norm{x-x'}
    \notag\\
    &\le
    8\theta k\lambda\cdot\OPT.
    \label{eq:aspect-original-expectation}
\end{align}
In addition, for every $x,y\in X$ we have
\begin{align}
    &\E\left[
        \mathbf{1}_{(x,y)\in S}
        \frac1{|P_x^*|}
        \sum_{x'\in P_x^*}\norm{f(x)-f(x')}
    \right]
    \notag\\
    &\quad\le
    \frac1{|P_x^*|}
    \sum_{x'\in P_x^*}
    \E\left[
        (1+\varepsilon)\mathbf{1}_{(x,y)\in S}\norm{x-x'}
    \right]
    \notag\\
    &\qquad
    +
    \frac1{|P_x^*|}
    \sum_{x'\in P_x^*}
    \E\left[
        \max\{\norm{f(x)-f(x')}-(1+\varepsilon)\norm{x-x'},0\}
    \right]
    \notag\\
    &\quad\le
    3(1+\varepsilon)\theta
    \frac1{|P_x^*|}
    \sum_{x'\in P_x^*}\norm{x-x'}.
    \label{eq:aspect-image-one-pair}
\end{align}
Summing this over all $x,y$ and using
$|P_x^*|\ge n/(2k)$ yields
\[
    \E\left[
        \sumbad
        \frac1{|P_x^*|}
        \sum_{x'\in P_x^*}\norm{f(x)-f(x')}
    \right]
    \le
    12(1+\varepsilon)\theta k\lambda\cdot\OPT.
\]
The same bounds hold for the symmetric $y$-terms, and therefore
\[
    \E[\mathcal D]
    \le
    40(1+\varepsilon)\theta k\lambda\cdot\OPT.
\]
By Markov's inequality,  $\Pr \bigsqpar{\mathcal D\le400(1+\varepsilon)\theta k\lambda\cdot\OPT}\geq 9/10.$ By a union bound, with probability at least $4/5$, $\Lambda$
occurs and the above bound on $\mathcal D$ holds, proving the lemma.
\end{proof}

\begin{lemma}[Balanced Partition with Near-Optimal Cost]
\label{lem:balanced-partition-approximation}
Let $X$ be an instance of min-sum $k$-clustering with aspect ratio at most $\Delta$.  Let
$\CALP=\{P_1,\ldots,P_k\}$ be a partition such that
    $\underset{i\in [k]}{\min} \,|P_i|\ge \frac{n}{2k}$. Then
\[
    \cost(\CALP)\le 4k^2\Delta\,\OPT.
\]
\end{lemma}

\begin{proof}
Normalize the minimum positive distance in $X$ to be $1$.  Then the maximum
distance is at most $\Delta$.  Let $\CALP^{\mathrm{opt}}$ be an optimal
partition.  Some cluster of $\CALP^{\mathrm{opt}}$ has size at least $n/k$, and
therefore the optimum pays for at least
\[
    \frac{n}{k}\left(\frac{n}{k}-1\right)
    \ge
    \frac{n^2}{4k^2}
\]
positive-distance pairs, assuming $n\ge 2k$ (the proof follows trivially otherwise).  Thus
\[
    \OPT\ge \frac{n^2}{4k^2}.
\]
At the same time, every partition has at most $n^2$ ordered pairs inside its
clusters, and each pair contributes at most $\Delta$.  Hence $\cost(\CALP)\le n^2\Delta\le 4k^2\Delta\,\OPT.$
\end{proof}

\begin{proof}[Proof of \Cref{thm:main-jl-aspect-ratio}]
Let
$\varepsilon'=(1+\varepsilon)^{1/3}-1=O(\varepsilon)$ and $\,\lambda=4k^2\Delta.$
Choose
\[
    \theta
    =
    \min\left\{
        \frac{(\varepsilon')^2}{20},
        \frac{\varepsilon'}{800(1+\varepsilon')^3\lambda k^2}
    \right\},
\]
and let
\[
    \delta\le
    \min\left\{\frac{\theta^7}{600},\frac1{10k^3}\right\},
    \qquad
    \rho\le\theta.
\]
By \Cref{lem:gaussian-set-pair-preservation}, there is a sufficiently large
constant $c>0$ such that, a Gaussian JL map $f:\R^d\to\R^t$ with $t\ge c\varepsilon^{-2}\bigpar{\log(1/\varepsilon)+\log k+\log\Delta}$ is an $(\varepsilon',\delta)$-random dimension reduction.
Additionally, by \Cref{lem:gaussian-one-sided-excess},  $f$ is an $(\varepsilon',\delta,\rho)$-random dimension reduction.

Take any balanced partition $\CALP^*$ satisfying
$\min_i |P_i^*|\ge n/(2k)$.  By
\Cref{lem:balanced-partition-approximation}, it achieves $\lambda$-approximation. Apply \Cref{lem:aspect-random-subset-preservation} with $\varepsilon'$. With the above choice of $\theta$, setting
$\Gamma=400(1+\varepsilon')\theta k\lambda\cdot\OPT$ gives
$A\Gamma\le \varepsilon'\OPT/(2k)$, as required by
\Cref{lem:bad-cluster-extraction}.
Applying that lemma and rescaling constants proves the theorem.
\end{proof}

\subsection{Generalization to power $z$}
\label{sec:l2p-generalization}
We briefly record how the dimension-reduction argument extends from
Euclidean distances to their \(z\)-th powers. Fix \(z\ge 1\), and define
\[
    \cost_z(\calP)
    =
    \sum_i\sum_{x,y\in P_i}\norm{x-y}^z.
\]
Write $\OPT_z$ for the optimal value under $\cost_z$.
We shall use that for all $a,b\ge0$ and $0<\gamma\le1$,
\begin{equation}
    (a+b)^z\le(1+\gamma)a^z+
    \left(\frac{2z}{\gamma}\right)^{z-1}b^z,
    \qquad
    (a+b)^z\le2^{z-1}(a^z+b^z). \label{eq:l2p-inequalities}
\end{equation}

\paragraph{Doubling dimension.}
We extend \Cref{thm:main-jl-ddim} to preserve $\cost_z(\calP)$
within a factor of $1\pm\varepsilon$, simultaneously for all
$k$-partitions $\calP$.
In its proof, use $\varepsilon'=\Theta(\varepsilon/z)$ as the
distortion parameter.
Apply \Cref{lem:gaussian-pair-tail} with this parameter in the
derivation of \eqref{eq:ddim-net-preservation}.
Then $(1\pm\varepsilon')^z=1\pm O(\varepsilon)$.

Choose each reference center $c_j^*$ to minimize
$\sum_{y\in P_j^*}\norm{y-c_j^*}^z$, and replace $\OPT$ by
$\OPT_z^{1/z}$ in the radii and in
\eqref{eq:ddim-radius-bound}--\eqref{eq:ddim-radius-sum}.
Use $(\varepsilon' r_i/100)$-nets and apply \Cref{lem:gaussian-local-ball-tail} to balls of radius
$\varepsilon' r_i/100$ to obtain \eqref{eq:ddim-net-ball}
with parameter $\varepsilon'$. 
Applying \eqref{eq:l2p-inequalities} to the inequalities
in \eqref{eq:ddim-rounding-original} and
\eqref{eq:ddim-rounding-image} gives the corresponding bounds
for $z$-th powers. In particular,
\[
\begin{aligned}
    \norm{x-y}^z
    &\le(1+\gamma)\norm{u(x,y)-u(y,x)}^z\\
    &\quad+
    \left(\frac{4z}{\gamma}\right)^{z-1}
    \bigl(\norm{x-u(x,y)}^z+\norm{y-u(y,x)}^z\bigr),
\end{aligned}
\]
And also
\[
\begin{aligned}
    \norm{u(x,y)-u(y,x)}^z
    &\le (1+\gamma)\norm{x-y}^z\\
    &\quad+
    \left(\frac{4z}{\gamma}\right)^{z-1}
    \bigl(\norm{x-u(x,y)}^z+\norm{y-u(y,x)}^z\bigr).
\end{aligned}
\]
Take $\gamma=\varepsilon/10$ and in the definitions of $g$ and $g_f$ in the proof of
\Cref{lem:ddim-random-subset-preservation}, raise each norm
to its $z$-th power, including the norms inside the averages.
Denote the resulting costs by $\cost_{g,z}(C)$ and
$\cost_{g,z}(f(C))$.

Summing the above bounds  over all pairs in $C$ bounds
$\cost_z(C)$ in terms of $\cost_{g,z}(C)$ and
$\cost_{g,z}(C)$ in terms of $\cost_z(C)$.
The corresponding bounds hold for
$\cost_z(f(C))$ and $\cost_{g,z}(f(C))$.
The inequality $(a+b)^z\le 2^{z-1}(a^z+b^z)$ introduces
a factor $2^{z-1}$ in \eqref{eq:distorted_cost_bound}.
Consequently, the analogues of
\eqref{eq:ddim-cost-g-forward}--\eqref{eq:ddim-cost-g-backward}
for $\cost_{g,z}$ have multiplicative factor
$(1+\varepsilon)(1+2^{z+2}\theta)$, after rescaling constants.
We therefore add the condition
$\theta\le\varepsilon/(100\cdot2^{z+2})$
to the choice of $\theta$ in the proof of
\Cref{thm:main-jl-ddim}.
The application of \Cref{thm:sparse-graph} is unchanged.

Combining the $z$-th power versions of
\eqref{eq:ddim-cost-g-forward}--\eqref{eq:ddim-cost-g-backward}
and \eqref{eq:ddim-compare-one}--\eqref{eq:ddim-compare-four}
gives an upper bound on $\cost_z(C)$ in terms of
$\cost_z(f(C))$ and an upper bound on $\cost_z(f(C))$
in terms of $\cost_z(C)$, each with multiplicative factor
$1+O(\varepsilon)$ and additive term
$O\!\left(
    \varepsilon\sum_{(x,y)\in S}\norm{y-c^*(y)}^z
    +\varepsilon\OPT_z/n^8
\right)$.
Applying Markov's inequality to the $z$-th power version of
\eqref{eq:ddim-charge-expectation} then gives
\eqref{eq:ddim-subset-first} and \eqref{eq:ddim-subset-second}
with $\cost_z$ and $\OPT_z$, with probability at least $4/5$.
Thus the hypotheses of \Cref{lem:bad-cluster-extraction}
are satisfied, and its proof applies to $\cost_z$.

Since the pairwise distance estimates use
$\varepsilon'=\Theta(\varepsilon/z)$,
\Cref{lem:gaussian-pair-tail} gives the leading factor $z^2$.
The net size contributes $\ddim\log(z/\varepsilon)$, while
the additional restriction on $\theta$ contributes $O(z)$
to $\log(1/\delta)$ through $\delta\le\theta^7/600$.
Together with \Cref{lem:gaussian-local-ball-tail},
this gives the target dimension
$
    t=O\!\left(
        z^2\varepsilon^{-2}
        \bigpar{\ddim\log(z/\varepsilon)+\log k+\log\log n+z}
    \right).
$

We next consider \Cref{thm:main-jl-aspect-ratio}, showing that
its guarantees also extend to preserving $\cost_z(\calP)$
within a factor of $1\pm\varepsilon$, simultaneously for all
partitions $\calP$ of $X$ into $k$ clusters.
\paragraph{Aspect ratio.}
For the proof of \Cref{thm:main-jl-aspect-ratio}, we require
that, for every fixed $A,B\subseteq X$, with probability at least $1-\delta$, 
\[
    \sum_{x\in A,\,y\in B}\norm{f(x)-f(y)}^z
    \in
    (1\pm\varepsilon)
    \sum_{x\in A,\,y\in B}\norm{x-y}^z.
\]

To this end we apply \cite[Theorem~3.5]{MakarychevMR19} with one cluster,
distortion parameter $\varepsilon'$, and failure probability
$\delta$, to the multiset
$Z_{A,B}=\{\pm(x-y):x\in A,\,y\in B\}$.
By symmetry and convexity, the origin is an optimal center
both before and after applying $f$, so the one-cluster costs
are twice the two sums above.
Thus it suffices to take
$ \label{eq:l2p-aspect-set-pair-dimension}
    t=O\!\left(
        z^2\varepsilon^{-2}
        \bigpar{\log(1/\delta)+z\log(z/\varepsilon)+z^2}\right).$

We also need that, for every $x,y\in X$,
\begin{equation}
    \E\!\left[
        \max\left\{
            \norm{f(x)-f(y)}^z
            -(1+\varepsilon)\norm{x-y}^z,0
        \right\}
    \right]
    \le \rho\norm{x-y}^z.\label{eq:l2p-aspect-excess}
\end{equation}
To this end, we apply \cite[Definition~2.1, Lemma~C.1, and Claim~C.2]{MakarychevMR19} with parameter $\varepsilon'$ and use $(1+\varepsilon')^z\le 1+\varepsilon$, hence it suffices to take
$
\label{eq:l2p-aspect-excess-dimension}
    t=O\!\left(
        z^2\varepsilon^{-2}\bigpar{z+\log(1/\rho)}
    \right).
$

Use $\varepsilon'$ in \Cref{def:distortion-graph} and apply
\Cref{lem:gaussian-pair-tail} with this parameter.
The proof of \Cref{lem:balanced-partition-approximation}
applies to $\cost_z$ with $\lambda_z=4k^2\Delta^z$.

The inequalities in \eqref{eq:l2p-inequalities} give the
following $z$-th power version of
\eqref{eq:aspect-original-rounding}:
\begin{equation}
\label{eq:l2p-aspect-rounding}
\begin{aligned}
    &\left|
        \norm{x-y}^z
        -\frac{1}{|P_x^*||P_y^*|}
        \sum_{x'\in P_x^*,\,y'\in P_y^*}
            \norm{x'-y'}^z
    \right|\\
    &\quad\le
    \gamma\norm{x-y}^z
    +\left(\frac{4z}{\gamma}\right)^{z-1}
    \left(
        \frac1{|P_x^*|}
        \sum_{x'\in P_x^*}\norm{x-x'}^z
        +
        \frac1{|P_y^*|}
        \sum_{y'\in P_y^*}\norm{y-y'}^z
    \right).
\end{aligned}
\end{equation}
Define
\[
\begin{aligned}
    \mathcal{D}_z=\sum_{(x,y)\in S}\Bigg(
    &\frac1{|P_x^*|}\sum_{x'\in P_x^*}
        \bigl(\norm{x-x'}^z+\norm{f(x)-f(x')}^z\bigr)\\
    &+\frac1{|P_y^*|}\sum_{y'\in P_y^*}
        \bigl(\norm{y-y'}^z+\norm{f(y)-f(y')}^z\bigr)
    \Bigg).
\end{aligned}
\]
Choosing the parameter $\rho$ in \eqref{eq:l2p-aspect-excess}
to satisfy $\rho\le\theta$, the $z$-th power calculations in
\eqref{eq:aspect-original-expectation} and
\eqref{eq:aspect-image-one-pair} give
\[
    \E[\mathcal D_z]
    \le 40(1+\varepsilon)\theta k\lambda_z\cdot\OPT_z.
\]
Hence, by Markov's inequality,
$\mathcal D_z\le400(1+\varepsilon)\theta k\lambda_z\cdot\OPT_z$
with probability at least $9/10$.
The inequality $(a+b)^z\le2^{z-1}(a^z+b^z)$ introduces
a factor $2^{z-1}$ in \eqref{eq:distorted_cost_bound}.
Consequently, the $z$-th power versions of
\eqref{eq:aspect-g-forward} and \eqref{eq:aspect-g-backward}
hold with multiplicative factor
$(1+2^{z+2}\theta)/(1-\varepsilon)$.
Summing \eqref{eq:l2p-aspect-rounding} over $S$ for the
original points and their images under $f$, and using these versions
of \eqref{eq:aspect-g-forward} and
\eqref{eq:aspect-g-backward}, gives
\[
\begin{aligned}
    \cost_z(C)
    &\le A\bigl(\cost_z(f(C))+\Gamma\bigr),\\
    \cost_z(f(C))
    &\le A\bigl(\cost_z(C)+\Gamma\bigr),
\end{aligned}
\]
where
\[
    A=\frac{(1+\gamma)(1+2^{z+2}\theta)}
            {(1-\gamma)(1-\varepsilon)},
    \qquad
    \Gamma=
    400(1+\varepsilon)
    \left(\frac{4z}{\gamma}\right)^{z-1}
    \lambda_z\theta k\OPT_z.
\]
For every random $C\subseteq X$ that may depend on $f$,
these inequalities hold with probability at least $4/5$:
the bound on $\mathcal D_z$ fails with probability at most
$1/10$, and preservation of the $k^2$ reference-cluster
sums fails with probability at most $k^2\delta\le1/10$.
Accordingly, in the proof of
\Cref{thm:main-jl-aspect-ratio}, choose
\[
    \theta=\min\left\{
        \frac{\varepsilon}{100\cdot2^{z+2}},
        \frac{\varepsilon}
             {800(1+\varepsilon)^3\lambda_z k^2}
        \left(\frac{\gamma}{4z}\right)^{z-1}
    \right\},
    \qquad
    \delta=\min\left\{\frac{\theta^7}{600},\frac1{10k^3}\right\},
    \qquad
    \rho=\theta.
\]
With $\gamma=\varepsilon/10$, these choices give
$A\le(1+\varepsilon)^2$ and
$A\Gamma\le\varepsilon\OPT_z/(2k)$.
Thus \Cref{lem:bad-cluster-extraction} applies to $\cost_z$.
These parameters satisfy
$\log(1/\delta)
=O\bigl(z\log(z/\varepsilon)+\log k+z\log\Delta\bigr)$.
This gives the target dimension $
    t=O\!\left(
        z^4\varepsilon^{-2}
        \bigpar{\log(z/\varepsilon)+\log k+z\log\Delta}
    \right).$
\subsection{Lower Bounds}

We next prove lower bounds for the two parameters that appear in the dimension
reduction theorems above.  The first construction rules out a target dimension that is independent of $\ddim(X)$.  The second construction rules out a constant target
dimension when $\ddim(X)=O(1)$ and the number of distance scales grows.
We allow multisets for convenience, as multiplicities can be simulated by placing distinct points arbitrarily close together.

\begin{theorem}[Lower Bounds for For-All Preservation]
\label{thm:lower-bounds}
There is a universal constant $c>0$ such that the following hold.
\begin{enumerate}
    \item For infinitely many $n$, there is a set $X\subseteq\R^d$ of size
    $n$ and doubling dimension $\ddim(X)=\Theta(\log n)$ such that a Gaussian JL map with target dimension
    $t\le c\varepsilon^{-2}\ddim(X)$ fails, with constant probability, to preserve the min-sum $2$-clustering
    cost of all partitions of $X$.
    \item For infinitely many $n$, there is a set $X\subseteq\R^d$ of size
    $n$ and constant doubling dimension such that a Gaussian JL map with target dimension
    $t\le c\varepsilon^{-2}\log\log n$ fails, with constant
    probability, to preserve the min-sum $2$-clustering cost of all partitions of $X$.
\end{enumerate}
\end{theorem}

\begin{proof}
We first give the lower bound showing that a dependence on $\ddim(X)$ is
necessary.  Let $d$ be large and set $k=2$.  Define the multiset
\[
    X=\{0_d\}^n\cup\{e_1,\ldots,e_{\sqrt n}\},
\]
where $\{e_i\}$ are standard basis vectors in $\R^d$.
The distinct points of $X$ have doubling dimension $\Theta(\log n)$: the
$\sqrt n$ basis vectors are a constant-scale packing inside a constant-radius
ball, and the total number of distinct points is $\sqrt n+1$.

Consider the partition
\[
    C_1^*=\{0_d\}^n,
    \qquad
    C_2^*=\{e_1,\ldots,e_{\sqrt n}\}.
\]
The cost satisfies
\[
    \cost(C_2^*)
    =
    \sum_{i\ne j}\norm{e_i-e_j}
    =
    \sqrt2\,\sqrt n(\sqrt n-1)
    \le
    \sqrt2\,n,
\]
while $\cost(C_1^*)=0$.  Suppose, toward a contradiction, that there is another
optimal partition $\{C_1,C_2\}$.  Then some $e_i$ is clustered with at least one
copy of $0_d$.  Without loss of generality, assume $e_i\in C_1$ and all zeros
are in $C_1$, since moving all zeros into $C_1$ only lowers the cost.  If there
is another $e_j\in C_1$ with $j\ne i$, then $\cost(C_1)\ge 2n$, contradicting
optimality.  Hence $C_1=\{0_d\}^n\cup\{e_i\}$ and
$C_2=\{e_j:j\ne i\}$.  But then
\[
    \cost(C_1)=2n,
    \qquad
    \cost(C_2)=\sqrt2(\sqrt n-1)(\sqrt n-2)>n,
\]
so the total cost exceeds $2n$, again a contradiction.  Thus the displayed
partition is the unique optimum.

Let $\pi:\R^d\to\R^t$ be a JL projection with $O(\varepsilon^{-2}\log(1/\varepsilon))
    \le    t
    \le
    c\varepsilon^{-2}\ddim(X)$ for a sufficiently small universal constant $c>0$.  Since
$\ddim(X)=\Theta(\log n)$, the choice of $c$ will make the lower-tail
probability below large enough.
Set \(\eta=40\varepsilon\).  By \cite[Lemma~6.3]{JKSSY26}, for a
$\chi_t^2$ random variable $X_t$,
\[
    \Pr\!\left[X_t<\frac{t}{1+\eta}\right]
    \ge
    \frac{e^{-C_0\eta^2t}}{t}.
\]
Choosing the universal constant \(c\) sufficiently small gives
\(C_0\eta^2t\le \frac1{100}\log n\).  Taking \(n\) sufficiently large as a
function of \(\varepsilon\), the right-hand side is then at least
\(10/\sqrt n\).  Since
\[
    \sqrt{\frac{1}{1+\eta}}\le 1-10\varepsilon
    \qquad\text{for }0<\varepsilon<1/40,
\]
we obtain, for each $i$,
\[
    \Pr\bigl[\norm{\pi e_i}\le 1-10\varepsilon\bigr]
    \ge
    \frac{10}{\sqrt n}.
\]
By independence,
\[
    \Pr\Bigl[
        \forall i\in[\sqrt n],\
        \norm{\pi e_i}>1-10\varepsilon
    \Bigr]
    \le
    \frac1{10}.
\]
Thus, with probability at least $9/10$, there exists $i^*$ such that
\begin{equation}
\label{eq:general-short-vector}
    \norm{\pi e_{i^*}}\le 1-10\varepsilon.
\end{equation}
Additionally, by \cite[Lemma~2.1.9]{Danos21}, with probability at least
$9/10$,
\begin{equation}
\label{eq:general-pairwise-preservation}
    \sum_{i\ne j}\norm{\pi(e_i)-\pi(e_j)}
    \in
    (1\pm\varepsilon)
    \sum_{i\ne j}\norm{e_i-e_j}.
\end{equation}By a union bound, both events occur with probability at
least $8/10$.

Consider the partition
\[
    C_1=\{0_d\}^n\cup\{e_{i^*}\},
    \qquad
    C_2=\{e_i:i\ne i^*\}.
\]
Its original cost is
\[
    \cost(C_1,C_2)
    =
    (2+\sqrt2)n-3\sqrt2\sqrt n+O(1).
\]
After projection, using \eqref{eq:general-short-vector} and
\eqref{eq:general-pairwise-preservation},
\begin{align*}
    \cost(\pi(C_1),\pi(C_2))
    &\le
    2n(1-10\varepsilon)
    +(1+\varepsilon)\sqrt2\sqrt n(\sqrt n-1)
    \\
    &\le
    (2+\sqrt2)n-15\varepsilon n-4\sqrt2\sqrt n
    \\
    &<
    (1-\varepsilon)\cost(C_1,C_2),
\end{align*}
for $\varepsilon\in(3\sqrt2/\sqrt n,1/40)$.  Thus the clustering cost is not
preserved.

For the constant-doubling lower bound, let $m\le d$ and define
\[
    X
    =
    \{2^i e_i:\ i\in[m]\}
    \cup
    \{0_d\}^{m^3 2^m}.
\]
This set has constant doubling dimension and satisfies
$\log m=\Theta(\log\log |X|).$
Indeed, the nonzero points lie on exponentially separated scales, and every
ball is covered by a constant number of balls of half the radius; adding
repeated copies of the origin does not change the doubling dimension.
As before, applying the same tail estimate with \(\eta=40\varepsilon\), for
\[
    O(\varepsilon^{-2}\log(1/\varepsilon))
    \le
    t
    \le
    c\varepsilon^{-2}\log m,
\]
where \(c>0\) is sufficiently small.  Taking \(m\) sufficiently large as a
function of \(\varepsilon\), the lower-tail probability for each direction is at
least \(10/m\), and therefore with probability at least $9/10$ there exists
\(i^*\) such that $\norm{\pi e_{i^*}}\le 1-10\varepsilon,$
and with probability at least $9/10$,
\[
    \sum_{i\ne j}
    \norm{
        \pi\bigl(2^i e_i\bigr)
        -
        \pi\bigl(2^j e_j\bigr)
    }
    \in
    (1\pm\varepsilon)
    \sum_{i\ne j}
    \norm{2^i e_i-2^j e_j}.
\]
Consider the partition
\[
    C_1=\{0_d\}^{m^3 2^m}\cup\{2^{i^*}e_{i^*}\},
    \qquad
    C_2=\{2^i e_i:\ i\ne i^*\}.
\]
Let
\[
    A=m^3 2^m\cdot 2^{i^*},
    \qquad
    T=
    \sum_{i,j\ne i^*}
    \norm{2^i e_i-2^j e_j}.
\]
Then $\cost(C_1,C_2)=2A+T$.  Moreover,
\[
    S:=\sum_{i\ne j}\norm{2^i e_i-2^j e_j}
    =
    O(m2^m).
\]
Since $A\ge 2m^3 2^m$ and $\varepsilon>1/m$, for $m$ sufficiently large we have
$S\le \varepsilon A$.  Thus the contribution of the cluster containing the
origin dominates the cost of the second cluster.
After projection,
\begin{align*}
    \cost(\pi(C_1),\pi(C_2))
    &\le
    2(1-10\varepsilon)A
    \\
    &\quad
    +(1+\varepsilon)
    S
    \\
    &\le
    2(1-8\varepsilon)A
    \\
    &<
    (1-\varepsilon)(2A+T)
    =
    (1-\varepsilon)\cost(C_1,C_2),
\end{align*}
for $\varepsilon\in(1/m,1/40)$.  This proves the second item.
\end{proof}

One may hope to bypass these lower bounds by preserving only a restricted
family of solutions, for instance the \(O(1)\)-approximate partitions, while
arguing that all other partitions remain too expensive after the projection (to not
affect the projected optimum).  This relaxation does not remove the obstruction:
there are constant-doubling variants of the examples above in which, with
constant probability, even an \(O(1)\)-approximate partition has its min-sum cost
distorted by a \(1\pm\varepsilon\) factor.  Thus the difficulty is not only that
irrelevant high-cost partitions may become cheap after projection, but rather the family of
near-optimal partitions might not be multiplicatively preserved.  We omit the
details.

\section{Movement Coresets for Min-Sum \texorpdfstring{$k$}{k}-Clustering}
\label{sec:movement-coresets}
We prove a density-adaptive movement-coreset construction for doubling metrics. The proof first develops a fast annular construction, which is then used as an intermediate reduction before applying the density-adaptive movement scales.

\DensityAdaptiveCoreset*

We use the movement-coreset terminology from the introduction. Throughout this
section, unless stated otherwise, once a representative set \(S\) is
constructed, \(\pi\) denotes the nearest-neighbor movement rule induced by
\(S\), with ties broken by a fixed rule, and the weights are
\(w(s)=|\pi^{-1}(s)|\) for every \(s\in S\).

For a weighted set \((S,w)\) and a partition $\calQ$ of $S$, define
\[
    \cost_{S,w}(\calQ)
    =
    \sum_{Q\in\calQ}\sum_{s,t\in Q}w(s)w(t)\dist(s,t),
\]
and, once a movement rule \(\pi:X\to S\) is fixed, lift $\calQ$ to \(X\)
by assigning each $x\in X$ to the part of $\calQ$ containing $\pi(x)$.
For a weighted instance \((Y,w)\), let
\[
    \OPT_{Y,w}
    =
    \min_{\calQ}\cost_{Y,w}(\calQ),
\]
where the minimum is over all partitions of \(Y\) into at most \(k\) parts.

We will also use the balanced \(k\)-median view of min-sum.  For a weighted
cluster \(A\) with weights \(w_x\), write
\[
    W(A)=\sum_{x\in A}w_x
\]
and define its balanced \(k\)-median cost by
\[
    \cost_{\mathrm{bkm},w}(A)
    =
    W(A)\min_{p\in A}\sum_{x\in A}w_x\dist(x,p).
\]
The unweighted case is defined by setting every weight equal to \(1\).
For a partition, \(\cost_{\mathrm{bkm},w}\) denotes the sum of these
cluster costs.
This is equivalent, up to a factor \(2\), to the weighted min-sum cost
\[
    \sum_{x,y\in A}w_xw_y\dist(x,y).
\]
Indeed, by averaging over \(p\in A\) with probability \(w_p/W(A)\),
\[
    \cost_{\mathrm{bkm},w}(A)
    \le
    \sum_{x,y\in A}w_xw_y\dist(x,y),
\]
and by the triangle inequality, for every \(p\in A\),
\[
    \sum_{x,y\in A}w_xw_y\dist(x,y)
    \le
    2W(A)\sum_{x\in A}w_x\dist(x,p).
\]
Thus, whenever we pass between min-sum and balanced \(k\)-median below,
we use these two inequalities.

We will also use the following bicriteria convention throughout the section.
Some of the approximation algorithms we invoke may return \(O(k)\) clusters
rather than exactly \(k\).  This is sufficient for our purposes: replacing
\(k\) by \(O(k)\) in the constructions and analyses only changes constant
factors.  To keep the notation simple, we continue to write the constructions using \(k\) parts or \(k\) centers.

\subsection{A Fast Annular Coreset}
We begin with an annular coreset construction, which will later be used to reduce the input size before applying the density-adaptive construction.
\begin{restatable}[Fast annular coreset]{lemma}{FastAnnularCoreset}
\label{thm:annular-coreset}
Let $(X,\dist)$ be an $n$-point metric space of doubling dimension at most
$\ddim$.  For every $0<\varepsilon<\frac{1}{10}$, one can construct an
\(\varepsilon\)-movement coreset \((S,w)\) with \(S\subseteq X\) and
\[
    |S| \le \left(\frac{k}{\varepsilon}\right)^{O(\ddim)}\log n .
\]
The associated movement rule is the nearest-neighbor map. Moreover, the
construction can be implemented in randomized near-linear time.
\end{restatable}
\subsubsection{Construction}
We start by showing the construction, assuming we have a \(c\)-approximate
reference solution for the balanced \(k\)-median formulation.  Let $\calR=\{R_1,\ldots,R_k\}$
be the partition induced by this solution, let \(c_i\) be the center serving
\(R_i\), and let
\[
    \widehat{\OPT}
    =
    \sum_{i=1}^k |R_i|\sum_{x\in R_i}\dist(x,c_i)
    \le c\OPT .
\]
Let
\[
    r_0=20\widehat{\OPT},
    \qquad
    r_j=r_0/2^j,
    \qquad
    j=0,\ldots,L,
\]
where $L=\lceil\log_2(40n^{10})\rceil$.  Let $\eta$ be a parameter to be
chosen later.  For every $i\in[k]$ and every $j$,
take a maximal $\eta r_j$-packing of
\[
    X\cap B(c_i,r_j).
\]
Let $S$ be the union of these packings.

\begin{lemma}[Annular movement]
\label{lem:annular-movement}
For every $x\in X$,
\[
    \dist(x,\pi(x))
    \le
    2\eta
    \max\left\{
        \min_{i\in[k]}\dist(x,c_i),
        \frac{\widehat{\OPT}}{n^{10}}
    \right\}.
\]
\end{lemma}

\begin{proof}
Let $i$ be such that $x\in R_i$.  Since \(x\in R_i\), the term \(|R_i|\dist(x,c_i)\) is part of the
balanced \(k\)-median reference cost, thus
\[
    \dist(x,c_i)\le \widehat{\OPT}.
\]
Hence
$\min_i \dist(x,c_i)\le \widehat{\OPT}$ for every $x$.
If $\min_i \dist(x,c_i)\le \widehat{\OPT}/n^{10}$, then $x$ belongs to the finest
ball around the nearest reference median, and maximality of the packing gives
movement at most $O(\eta\widehat{\OPT}/n^{10})$.  Otherwise choose $j$ such that
\[
    r_{j+1}
    <
    \min_i \dist(x,c_i)
    \le
    r_j .
\]
The packing at scale $r_j$ around a nearest reference median contains a point
within distance $\eta r_j\le 2\eta\min_i \dist(x,c_i)$ from $x$.
\end{proof}

\begin{lemma}[Annular preservation]
\label{lem:annular-preservation}
Let $X$ have doubling dimension at most $\ddim$, and let $\calR$ be a
$c$-approximate partition.  For every $0<\varepsilon<1/10$, the
annular construction with a sufficiently small
\[
    \eta=\Theta\!\left(\frac{\varepsilon}{ck}\right)
\]
gives an $\varepsilon$-movement coreset.
\end{lemma}

\begin{proof}
Fix a partition $\calP$.  In each cluster, orient every unordered pair $\{x,y\}$ toward the endpoint whose cluster in $\calR$ is larger. Thus if the oriented pair is $(x,y)$ and $y\in R_\ell$, then the cluster containing $y$ has size at least the cluster containing $x$.

By \cref{lem:annular-movement}, for such a pair,
\[
    \dist(x,\pi(x))
    \le
    O(\eta)
    \left(
        \dist(x,c_\ell)+\frac{\widehat{\OPT}}{n^{10}}
    \right)
    \le
    O(\eta)
    \left(
        \dist(x,y)+\dist(y,c_\ell)+\frac{\widehat{\OPT}}{n^{10}}
    \right),
\]
and
\[
    \dist(y,\pi(y))
    \le
    O(\eta)
    \left(
        \dist(y,c_\ell)+\frac{\widehat{\OPT}}{n^{10}}
    \right).
\]
The triangle inequality therefore gives both
\begin{align}
\label{eq:annular-forward}
    \dist(x,y)
    &\le
    (1+O(\eta))\dist(\pi(x),\pi(y))
    +
    O(\eta)\dist(y,c_\ell)
    +
    O\!\left(\eta\frac{\widehat{\OPT}}{n^{10}}\right),
    \\
\label{eq:annular-backward}
    \dist(\pi(x),\pi(y))
    &\le
    (1+O(\eta))\dist(x,y)
    +
    O(\eta)\dist(y,c_\ell)
    +
    O\!\left(\eta\frac{\widehat{\OPT}}{n^{10}}\right).
\end{align}

It remains to sum the reference-median charge.  Fix $y\in R_\ell$.  The point
$y$ can be charged only by points whose reference cluster has size at most
$|R_\ell|$, and there are at most $k|R_\ell|$ such points.  Hence, over all
oriented pairs in all clusters,
\[
    \sum \dist(y,c_\ell)
    \le
    k\sum_{\ell=1}^k |R_\ell|\sum_{y\in R_\ell}\dist(y,c_\ell).
\]
By the definition of the balanced \(k\)-median reference cost,
\[
    \sum_{\ell=1}^k |R_\ell|\sum_{y\in R_\ell}\dist(y,c_\ell)
    =
    \widehat{\OPT}.
\]
Therefore the total charge is at most \(k\widehat{\OPT}\).
The floor term contributes at most
$O(\eta\widehat{\OPT}/n^8)$ after summing over all pairs.

Summing \eqref{eq:annular-forward} and \eqref{eq:annular-backward} over the
oriented pairs and using
\[
    \widehat{\OPT}\le c\OPT\le c\cost(\calP)
\]
shows that the total error is at most
\[
    O(\eta ck)\cost(\calP).
\]
Choosing the hidden constant in
$\eta=\Theta(\varepsilon/(ck))$ sufficiently small gives the desired
$\varepsilon$-movement guarantee.
\end{proof}

\begin{lemma}[Annular movement-charge bound]
\label{lem:annular-movement-charge}
Under the assumptions of \Cref{lem:annular-preservation}, the annular movement
rule \(\pi:X\to S\) can be made to satisfy, for every \(k\)-partition \(\calP\)
of \(X\),
\[
    \sum_{A\in\calP}|A|\sum_{x\in A}\dist(x,\pi(x))
    \le
    \varepsilon\cost(\calP),
\]
by choosing \(\eta=C_0\varepsilon/(ck)\)  for small enough constant $C_0>0$.
\end{lemma}

\begin{proof}
For every pair of points in a cluster, label the
endpoints so that \(y\in R_\ell\) belongs to the larger
reference cluster. By \Cref{lem:annular-movement},
\[
    \dist(x,\pi(x))+\dist(y,\pi(y))
    \le O(\eta)\left(
        \dist(x,y)+\dist(y,c_\ell)
        +\frac{\widehat{\OPT}}{n^{10}}
    \right).
\]
Summing over these pairs and using the charging argument
from the proof of \Cref{lem:annular-preservation} gives
\[
    \sum_{A\in\calP}(|A|-1)\sum_{x\in A}\dist(x,\pi(x))
    \le O(\eta)\cost(\calP)+O(\eta k\widehat{\OPT}).
\]
Moreover, \Cref{lem:annular-movement} and the definition of
\(\widehat{\OPT}\) imply
\[
    \sum_{x\in X}\dist(x,\pi(x))
    \le O(\eta\widehat{\OPT}).
\]
Thus
\[
    \sum_{A\in\calP}|A|\sum_{x\in A}\dist(x,\pi(x))
    \le O(\eta)\cost(\calP)+O(\eta k\widehat{\OPT})
    \le O(\eta ck)\cost(\calP),
\]
Choosing 
\(\eta=C_0\varepsilon/(ck)\)  for sufficiently small $C_0>0$ proves the claim.
\end{proof}

\noindent We next bound the size of the annular coreset by applying
\Cref{lem:doubling-packing-bound} to each net.
\begin{lemma}[Annular size]
\label{lem:annular-size}
The annular construction satisfies
\[
    |S|
    \le
    (k/\varepsilon)^{O(\ddim)}\log n .
\]
\end{lemma}

\begin{proof}
\Cref{lem:doubling-packing-bound} applied at each center and scale gives
\[
    |S|
    \le
    k(L+1)\eta^{-O(\ddim)}
    =
    k\eta^{-O(\ddim)}\log n .
\]
using $\eta$ from \Cref{lem:annular-preservation} gives the claim.
\end{proof}

\subsubsection{Obtaining a Constant Approximation Partition}
We now explain how to obtain the constant-factor reference partition needed by
the final coreset construction. The idea is to start from a $\mathrm{poly}(n)$-approximate
solution and improve it iteratively: in each iteration we build a small weighted
annular coreset around the current solution and run a weighted bicriteria
balanced \(k\)-median algorithm on this coreset. The next remark explains the
weighted bicriteria call used in each iteration; the following remark proves the
improvement recurrence.

\begin{remark}[Weighted bicriteria constant-factor approximation]
\label{rem:weighted-bkm}
The bicriteria algorithm of Bartal, Charikar and Raz \cite{BCR01} is stated for unweighted balanced \(k\)-median, where a cluster \(C\) served by a center \(p\) pays
\[
    |C|\sum_{x\in C}\dist(x,p).
\]
We use it on the weighted balanced \(k\)-median instances defined above, where the weights are multiplicities.

The only step of the argument that needs modification is the construction of the bipartite graph \(G'\). In the construction of Bartal, Charikar and Raz, each candidate center \(p\) has a copy \(p^{(r)}\) for every possible cluster size \(r\). The edge from a demand point \(x\) to \(p^{(r)}\) has cost \(r\dist(x,p)\), and \(p^{(r)}\) has capacity \(r\), namely
\[
    \sum_{x\mapsto p^{(r)}}1\le r .
\]
For a weighted instance, we replace the cluster-size parameter \(r\) by a total-weight parameter \(U\). Thus the copy \(p^{(U)}\) has edge cost \(Uw_x\dist(x,p)\) for a point \(x\) of weight \(w_x\), and capacity condition
\[
    \sum_{x\mapsto p^{(U)}}w_x\le U .
\]
For multiplicity weights, this is not a genuine change to the argument: it is exactly what one obtains by applying the unweighted construction to the multiset in 
which each point \(x\) appears \(w_x\) times, with the weighted notation recording these identical copies by their multiplicity. Hence the same bicriteria guarantee applies with cardinalities replaced by total weights.
The resulting solution may split copies of the same representative. This is
harmless for our application, since after converting its balanced
\(k\)-median guarantee to a min-sum guarantee,
\Cref{clm:unsplitting-representatives} removes all such splits without
increasing the min-sum cost. The resulting unsplit partition has balanced
\(k\)-median cost at most twice that of the original split solution.

For the running-time analysis, we make one further modification to the graph \(G'\). The parameter \(U\) is the capacity, in total weight, of a center copy. Since the weights in our coreset instances are multiplicities whose total sum is \(n\), all relevant exact capacities lie between \(1\) and \(n\). To keep the graph smaller, we allow only powers of two as possible values of \(U\), leaving \(O(\log n)\) values per candidate center.
This reduces the number of center copies and changes the guarantee only by a constant factor. Indeed, take any solution that uses exact total weights, and consider one of its clusters \(C\). If \(U\) is the smallest power of two with \(U\ge W(C)\), then
\[
    W(C)\le U<2W(C).
\]
Thus \(p^{(U)}\) has enough capacity to serve \(C\), and replacing \(W(C)\) by \(U\) can only increase this cluster's assignment cost, by at most a factor of \(2\):
\[
    W(C)\sum_{x\in C}w_x\dist(x,p)
    \le
    U\sum_{x\in C}w_x\dist(x,p)
    \le
    2W(C)\sum_{x\in C}w_x\dist(x,p).
\]
Applying this replacement to all clusters proves the constant-factor loss. Thus, if the weighted instance has point set \(Y\), then the two sides of \(G'\) have sizes \(|Y|\) and \(O(|Y|\log n)\).
\end{remark}

\begin{remark}[Bootstrapping a constant-factor reference partition]
\label{rem:fast-bootstrap} 
Let \(\rho>1\) be the constant approximation factor of the weighted balanced
\(k\)-median algorithm from \cref{rem:weighted-bkm}. Throughout this remark,
\(\OPT\) denotes the optimum min-sum \(k\)-clustering cost, whereas
\(\widehat{\OPT}_t\) denotes the balanced \(k\)-median reference value maintained
at round \(t\). We first obtain a polynomial-factor starting solution from any
\(n\)-approximation for ordinary \(k\)-median. For a center set \(C\), write
\[
    \cost_{\mathrm{med}}(C)
    =
    \sum_{x\in X} \dist(x,C),
\]
and let \(\OPT_{\mathrm{med}}\) be the optimum ordinary \(k\)-median cost over
all \(k\)-center sets.  Let \(C_0\) be an \(n\)-approximate \(k\)-median center
set, and let \(\calP_0=\{P_{0,1},\ldots,P_{0,k}\}\) be the induced nearest-center partition, with center \(c_{0,i}\) for each part. Then
\[
    \cost_{\mathrm{med}}(C_0)
    \le
    n\OPT_{\mathrm{med}}
    \le
    n\OPT.
\]
Here the second inequality follows because any min-sum \(k\)-clustering
induces an ordinary \(k\)-median solution of no greater cost.
For every cluster \(A\) of \(\calP_0\) served by a center \(p\),
\[
    |A|\sum_{x\in A}\dist(x,p)
    \le
    n\sum_{x\in A}\dist(x,p).
\]
Therefore the balanced \(k\)-median value of the initial reference
solution satisfies
\[
    \widehat{\OPT}_0
    =
    \sum_{i=1}^k |P_{0,i}|\sum_{x\in P_{0,i}}\dist(x,c_{0,i})
    \le
    n\cost_{\mathrm{med}}(C_0)
    \le
    n^2\OPT.
\]
Thus we may start the iterative improvement with
\(\widehat{\OPT}_0\le\alpha_0\OPT\), where \(\alpha_0\le n^2\).

We now show how to iteratively improve this bound on the balanced
\(k\)-median reference value. Suppose that at round \(t\) we have
\[
    \widehat{\OPT}_t\le \alpha_t\OPT.
\]
One round produces a new balanced \(k\)-median reference solution satisfying
\[
    \widehat{\OPT}_{t+1}
    \le
    O(\rho)\OPT+\frac12\,\widehat{\OPT}_t.
\]
Equivalently, this relative bound evolves according to
\[
    \alpha_{t+1}\le O(\rho)+\frac12\alpha_t.
\]
Thus the nonconstant part of the bound is halved in each round, and after
\(O(\log n)\) rounds the balanced \(k\)-median reference value becomes
\(O(\rho)\OPT\). We now describe one improvement round. Let
\(\calP_t=\{P_{t,1},\ldots,P_{t,k}\}\) be the current partition, with center
\(c_{t,i}\) for every part, and let
\[
    \widehat{\OPT}_t
    =
    \sum_{i=1}^k |P_{t,i}|\sum_{x\in P_{t,i}}\dist(x,c_{t,i}).
\]
Thus \(\widehat{\OPT}_t\) is the cost of the displayed feasible balanced
\(k\)-median solution; it is compared here with the optimum min-sum cost
\(\OPT\). Build the annular coreset around these centers using some parameter
\(\eta\).
The proof of \cref{lem:annular-preservation}, before the last step where
\(\widehat{\OPT}\le c\OPT\) is used, gives the following additive version: for
every \(k\)-partition \(\calP\),
\[
    \left|
    \cost(\calP)
    -
    \cost(\pi(\calP))
    \right|
    \le
    O(\eta)\cost(\calP)
    +
    O(\eta k\widehat{\OPT}_t).
\]
Both costs in this display are min-sum costs.
Choosing $\eta=\frac{1}{C\rho k}$ for a sufficiently large constant $C$ yields
\begin{equation}
    \left|
    \cost(\calP)
    -
    \cost(\pi(\calP))
    \right|
    \le
    \frac{1}{100\rho}\cost(\calP)
    +
    \frac{1}{100\rho}\widehat{\OPT}_t. \label{eq:additive-min-sum-bound}
\end{equation}
Denote the resulting coreset by \((S_t,w_t)\). Let
\(\widetilde{\calQ}_t\) be the possibly split
\(\rho\)-approximate weighted balanced \(k\)-median solution on
\((S_t,w_t)\). By \Cref{clm:unsplitting-representatives}, there is an
unsplit partition \(\calQ_t\) such that
\[
    \cost_{S_t,w_t}(\calQ_t)
    \le
    \cost(\widetilde{\calQ}_t).
\]
Let \(\calP_{t+1}\) be the lift of \(\calQ_t\) to \(X\). If
\(\calP^\star\) is an optimum min-sum partition of \(X\), then
\[
\begin{aligned}
    \cost(\pi(\calP_{t+1}))
    &\le
    \cost(\widetilde{\calQ}_t) \\
    &\le
    2\cost_{\mathrm{bkm},w_t}(\widetilde{\calQ}_t) \\
    &\le
    2\rho\cost_{\mathrm{bkm},w_t}(\pi(\calP^\star)) \\
    &\le
    2\rho\cost(\pi(\calP^\star)).
\end{aligned}
\]
The first and last inequalities are the two directions of the
balanced \(k\)-median/min-sum comparison, while the middle inequality is
the guarantee of the weighted balanced \(k\)-median algorithm. Applying
the additive bound in \eqref{eq:additive-min-sum-bound} first to \(\calP^\star\) and then to
\(\calP_{t+1}\) therefore yields
\[
    \cost(\calP_{t+1})
    \le
    O(\rho)\OPT
    +
    \frac{1}{4}\widehat{\OPT}_t.
\]

For every part \(P_{t+1,i}\), choose a point \(a_{t+1,i}\in P_{t+1,i}\)
such that
\[
    |P_{t+1,i}|
    \sum_{x\in P_{t+1,i}}\dist(x,a_{t+1,i})
    \le
    2\sum_{x,y\in P_{t+1,i}}\dist(x,y).
\]
The left-hand side is the contribution of the chosen center to the balanced
\(k\)-median reference value, whereas the sum on the right-hand side is the
min-sum cost of the part.
Such a point can be found with high probability by sampling
\(O(\log n)\) points from the part and retaining the one with minimum
total distance. Defining the next balanced \(k\)-median reference value
\(\widehat{\OPT}_{t+1}\) using these centers gives
\[
    \widehat{\OPT}_{t+1}
    \le
    2\cost(\calP_{t+1})
    \le
    O(\rho)\OPT+\frac12\widehat{\OPT}_t,
\]
which proves the claimed recurrence.

Set \(T=\lceil\log_2(n^2)\rceil\). Since \(n^2/2^T\le1\), iterating the
recurrence gives
\(\widehat{\OPT}_T\le O(\rho)\OPT\). This is a bound on the balanced
\(k\)-median reference value. To obtain a min-sum guarantee for its
partition, we use the other direction of the comparison:
\[
    \cost(\calP_T)
    \le
    2\cost_{\mathrm{bkm}}(\calP_T)
    \le
    2\widehat{\OPT}_T
    \le
    O(\rho)\OPT,
\]
so \(\calP_T\) is also an \(O(\rho)\)-approximation for min-sum
\(k\)-clustering.  One final annular construction with
\(\eta=\Theta(\varepsilon/(\rho k))\) then gives the coreset in
\cref{thm:annular-coreset}.  Each round still uses only \(O(k\log n)\) net
constructions.
\end{remark}

\subsubsection{Running Time}

A polynomial-factor starting solution is obtained in \(O(kn)\) time \cite{Gonzalez1985} \cite{MettuPlaxton2004}. In each improvement round, the annular construction is run with \(\eta=\Theta(1/(\rho k))\), so \cref{lem:annular-size} gives
\[
    |S_t|
    \le
    (\rho k)^{O(\ddim)}\log n .
\]
The annular construction uses \(O(k\log n)\) net computations. Using randomized net constructions in doubling metrics~\cite{HM06}, each net can be computed in \(2^{O(\ddim)}n\log n\) expected time. Hence the annular-construction part of one round takes \(2^{O(\ddim)}kn\log^2 n\) expected time.

Selecting the new reference centers adds \(O(n\log n)\) time per round:
for each part, we evaluate the total distance of \(O(\log n)\) sampled
candidates to all points in that part. Over \(T=O(\log n)\) rounds this
contributes \(O(n\log^2 n)\) time, which is dominated by the
annular-construction term above.

We next account for the weighted bicriteria call on this coreset. By \cref{rem:weighted-bkm}, the bipartite graph \(G'\) used in this call has \(|S_t|\) demand vertices and \(O(|S_t|\log n)\) center-copy vertices, and hence
\[
    |E(G')|=O(|S_t|^2\log n).
\]
The Bartal--Charikar--Raz construction then applies the uncapacitated facility-location algorithm of Jain and Vazirani, with modified edge costs. This algorithm runs in \(O(|E(G')|\log |E(G')|)\) time~\cite{jain2001approximation}, so one bicriteria call takes
\[
    \widetilde{O}(|S_t|^2\log n)
    =
    \widetilde{O}\!\left((\rho k)^{O(\ddim)}\log^3 n\right).
\]

There are \(T=\lceil\log_2(n^2)\rceil=O(\log n)\) improvement rounds. The final annular coreset construction has the same net-computation cost and does not require another call to the weighted approximation algorithm. Therefore the total expected running time is
\[
    2^{O(\ddim)}kn\log^3 n
    +
    \widetilde{O}\!\left((\rho k)^{O(\ddim)}\log^4 n\right).
\]
Apart from the displayed dependence on \(k\), \(\rho\), and \(\ddim\), the dependence on the original input size \(n\) is near-linear.
\subsubsection{Proof of the Annular Coreset Lemma}

\FastAnnularCoreset*
\begin{proof}
Assume first that we have a $c$-approximate reference partition.  Set
\[
    \eta=\frac{\varepsilon}{\beta\cdot c\cdot k}
\]
for a sufficiently large constant $\beta>0$.  The movement guarantee follows
from \cref{lem:annular-preservation}, and the size bound follows from
\cref{lem:annular-size}.  The required constant-factor reference partition and
the running-time bound are obtained by the bootstrapping procedure analyzed
above.
\end{proof}
\subsection{A Density-Adaptive Coreset}

The density-adaptive coreset removes the factor $k$ from the movement accuracy.  We first describe the construction assuming that we are given a value
\(\widehat{\OPT}\) satisfying
\[
    \OPT\le \widehat{\OPT}\le c\OPT
\]
for a constant \(c\).  We later remove this assumption by trying geometrically
decreasing estimates for \(\OPT\).

For $x\in X$ and $2\le m\le n$, let $\Gamma_m(x)$ be the set of $m$
points in $X$ that are closest to $x$, including $x$ itself, and define
\[
    \gamma_m(x)
    =
    \frac1m\sum_{v\in \Gamma_m(x)}\dist(x,v).
\]
Set
\[
    \rho(x)
    =
    \min_{2\le m\le n}
    \max\left\{
        \gamma_m(x),
        \frac{\widehat{\OPT}}{km^2}
    \right\}.
\]
The value $\rho(x)$ is the movement scale justified by the local density around
$x$.  The term $\gamma_m(x)$ pays for moving $x$ in a large cluster, while
$\widehat{\OPT}/(km^2)$ prevents too many well-separated points from appearing
at the same scale.

\subsubsection{Construction and Preservation}

Let
\[
    \tau=\frac{\varepsilon}{4(c+1)},
    \qquad
    r_{\min}=\frac{\widehat{\OPT}}{kn^2}.
\]
Since the second term in $\rho(x)$ is always at least $r_{\min}$, every point
satisfies $\rho(x)\ge r_{\min}$.  Define
\[
    H=\{x\in X:\rho(x)>\widehat{\OPT}\}.
\]
The points of $H$ are kept exactly, i.e added to $S$.  For all other points, let
\[
    J=\left\lceil\log_2(kn^2)\right\rceil,
    \qquad
    r_j=2^j r_{\min},
\]
and define the scale rings
\[
    R_j
    =
    \{x\in X\setminus H:\ r_j\le \rho(x)<2r_j\},
    \qquad
    0\le j\le J.
\]
For every nonempty ring $R_j$, choose a $(\tau r_j)$-net
$S_j\subseteq R_j$.  Set
\[
    S=H\cup\bigcup_{j=0}^J S_j.
\]
By the convention above, \(\pi\) is the nearest-neighbor movement rule.
Every point of \(H\) maps to itself.  For every \(x\in R_j\), the net \(S_j\)
contains a point within distance \(\tau r_j\) from \(x\), and therefore
\[
    \dist(x,\pi(x))<\tau r_j\le \tau\rho(x).
\]
We prove preservation in three steps.
First, \Cref{lem:movement-charge} bounds the total distance moved
by the points of a cluster, multiplied by its cardinality.
Next, \Cref{lem:one-cluster-preservation} uses the triangle inequality
to bound the resulting change in the cluster's min-sum cost.
The additive error per cluster is $2\tau\widehat{\OPT}/k$.
Finally, summing over at most $k$ clusters and using
$\widehat{\OPT}\le c\OPT$ then gives the preservation guarantee
in \Cref{lem:density-preservation}.

\begin{lemma}[Movement charge]
\label{lem:movement-charge}
Let $A\subseteq X$ have size $m\ge2$.  Suppose
\[
    \forall x\in A,
    \qquad
    \dist(x,\pi(x))\le \tau\rho(x).
\]
Then
\[
    m\sum_{x\in A}\dist(x,\pi(x))
    \le
    \tau\cost(A)
    +
    \frac{\tau\widehat{\OPT}}{k}.
\]
\end{lemma}

\begin{proof}
Fix $x\in A$.  The $m$ nearest points to $x$ are no farther in total than the
points of $A$, and hence
\[
    m\gamma_m(x)
    =
    \sum_{v\in \Gamma_m(x)}\dist(x,v)
    \le
    \sum_{y\in A}\dist(x,y).
\]
By definition of $\rho(x)$,
\[
    \rho(x)
    \le
    \gamma_m(x)+\frac{\widehat{\OPT}}{km^2}.
\]
Multiplying by $m$, summing over $x\in A$, and using
$\dist(x,\pi(x))\le\tau\rho(x)$ proves the lemma.
\end{proof}

\begin{lemma}[One cluster]
\label{lem:one-cluster-preservation}
Let $A\subseteq X$ have size $m\ge2$.  Then
\[
    \left|
    \cost(A)
    -
    \sum_{x,y\in A}\dist(\pi(x),\pi(y))
    \right|
    \le
    2\tau\cost(A)
    +
    \frac{2\tau\widehat{\OPT}}{k}.
\]
\end{lemma}

\begin{proof}
For every \(x,y\in A\), the triangle inequality gives
\[
    \left|\dist(x,y)-\dist(\pi(x),\pi(y))\right|
    \le
    \dist(x,\pi(x))+\dist(y,\pi(y)).
\]
Summing over all ordered pairs \(x,y\in A\), we get
\[
    \left|
    \cost(A)
    -
    \sum_{x,y\in A}\dist(\pi(x),\pi(y))
    \right|
    \le
    2|A|\sum_{x\in A}\dist(x,\pi(x)).
\]
The claim now follows from \cref{lem:movement-charge}.
\end{proof}

\begin{lemma}[Density-adaptive preservation]
\label{lem:density-preservation}
Let $(X,\dist)$ be a finite metric space of size $n$ and doubling dimension at most
$\ddim$.  Suppose
\[
    \OPT\le \widehat{\OPT}\le c\OPT .
\]
For every $0<\varepsilon<\frac{1}2{}$, the density-adaptive construction gives an
$\varepsilon$-movement coreset.
\end{lemma}

\begin{proof}
Let $\calP$ be any partition into at most $k$ parts.  Parts of size one have
zero cost and zero movement error.  Summing
\cref{lem:one-cluster-preservation} over all parts of size at least two gives
\[
    \left|
    \cost(\calP)
    -
    \sum_{P\in\calP}\sum_{x,y\in P}\dist(\pi(x),\pi(y))
    \right|
    \le
    2\tau\cost(\calP)+2\tau\widehat{\OPT}.
\]
Since
\[
    \widehat{\OPT}\le c\OPT\le c\cost(\calP),
\]
the right-hand side is at most
\[
    2\tau(c+1)\cost(\calP).
\]
With $\tau=\varepsilon/(4(c+1))$, this is at most
$\varepsilon\cost(\calP)$.
\end{proof}
The preceding argument gives more than preservation of the total
cost: it also bounds the sum of the movement terms used in the
triangle-inequality estimate. 
\begin{lemma}[Density-adaptive movement-charge bound]
\label{lem:density-movement-charge}
Under the assumptions of \Cref{lem:density-preservation}, the
density-adaptive movement rule \(\pi:X\to S\) satisfies, for every
\(k\)-partition \(\calP\) of \(X\),
\[
    \sum_{A\in\calP}|A|\sum_{x\in A}\dist(x,\pi(x))
    \le
    \varepsilon\cost(\calP).
\]
\end{lemma}
\begin{proof}
Write \(d_x=\dist(x,\pi(x))\).
When \(d_x>0\), we have \(\gamma_2(x)\le d_x/2\), so
\[
    d_x\le\tau\rho(x)
    \le\tau\left(\frac{d_x}{2}+\frac{\widehat{\OPT}}{4k}\right).
\]
This inequality also holds for \(d_x=0\).
Thus, for every $x\in X$, 
\(d_x\le\tau\widehat{\OPT}/(2k)\).
Combining this with \Cref{lem:movement-charge} and summing
over \(\calP\) yields
\[
    \sum_{A\in\calP}|A|\sum_{x\in A}d_x
    \le\tau\cost(\calP)+\tau\widehat{\OPT}
    \le\tau(c+1)\cost(\calP)
    \le\varepsilon\cost(\calP).
\]
\end{proof}

\subsubsection{Size Bound}

\begin{lemma}[High-scale points]
\label{lem:high-points}
The set $H$ has size at most $k$.
\end{lemma}

\begin{proof}
If $x\in H$, by definition of $\rho(x)$,
\[
    \widehat{\OPT}<\rho(x)
    \le
    \max\left\{
        \gamma_2(x),
        \frac{\widehat{\OPT}}{4k}
    \right\}.
\]
Hence $\gamma_2(x)>\widehat{\OPT}$, so the nearest neighbor of $x$ is farther
than $2\widehat{\OPT}$.  If $|H|>k$, then two points of $H$ lie in the same optimal cluster.  Their contribution to $\OPT$ is larger
$2\widehat{\OPT}\ge 2\OPT$, contradicting optimality.
\end{proof}

\begin{lemma}[Single scale ring]
\label{lem:one-ring}
For every scale ring $R_j$,
\[
    |S_j|
    \le
    3k\left(\frac{C}{\tau}\right)^{\ddim}
\]
for a universal constant $C>0$.
\end{lemma}

\begin{proof}
Fix a ring $R_j$ and write $r=r_j$ and $S_r=S_j$.  Let $N_r\subseteq S_r$ be a
$(16r)$-net of $S_r$.  Assign each $s\in S_r$ to its closest point in $N_r$, and
let $F_v$ be all the points assigned to $v\in N_r$.  Then $F_v \subseteq B(v,16r)$. 
Since $S_r$ is a $(\tau r)$-net, the points of $F_v$ are $\tau r$-separated and by
\Cref{lem:doubling-packing-bound},
\[
    |F_v|
    \le
    \left(\frac{C}{\tau}\right)^{\ddim}.
\]
Thus
\[
    |S_r|
    \le
    |N_r|\left(\frac{C}{\tau}\right)^{\ddim}.
\]
It remains to prove $|N_r|<3k$.

Let $T=|N_r|$.  For each $v\in N_r$, since $v\in R_j$, we have $\rho(v)<2r$.  Let
$m_v$ be the value of $m$ that attains the minimum defining $\rho(v)$.  Then
\[
    \gamma_{m_v}(v)<2r,
    \qquad
    \frac{\widehat{\OPT}}{km_v^2}<2r.
\]
The second inequality gives
\[
    m_v>\sqrt{\frac{\widehat{\OPT}}{2kr}},
\]
and the first implies that more than half of the points in $\Gamma_{m_v}(v)$ lie
inside $B(v,4r)$.

Define
\[
    h
    :=
        \left\lceil
        \sqrt{\frac{\widehat{\OPT}}{8kr}}
        \right\rceil.
\]
Since \(m_z>2\sqrt{\widehat{\OPT}/(8kr)}\), the preceding observation implies
that \(B(v,4r)\) contains at least \(h\) points.
For every $v\in N_r$, choose a set
\[
    G_v\subseteq B(v,4r),
    \qquad
    |G_v|=h.
\]
The choice of $h$ gives
\[
    rh^2\ge \frac{\widehat{\OPT}}{32k}
    \ge
    \frac{\OPT}{32k}.
\]

Distinct points of $N_r$ are at distance at least $16r$.  Therefore if
$x\in G_v$ and $y\in G_{v'}$ for $v\ne v'$, then
\[
    \dist(x,y)\ge 16r-4r-4r=8r.
\]
The sets $G_v$ are disjoint.

Let $O_1,\ldots,O_k$ be an optimal partition of $X$.  For $v\in N_r$ and $a\in[k]$, set
\[
    g_{va}=|G_v\cap O_a|,
    \qquad
    M_a=\sum_{v\in N_r}g_{va}.
\]
The number $Q$ of ordered pairs of selected points that lie in the same optimal
cluster but in different groups $G_v$ is
\[
    Q
    =
    \sum_{a=1}^k M_a^2
    -
    \sum_{v\in N_r}\sum_{a=1}^k g_{va}^2.
\]
Since $\underset{a}{\sum} M_a=\underset{v\in N_r}{\sum}|G_v|=Th$, Cauchy--Schwarz gives
\[
    \sum_{a=1}^k M_a^2
    \ge
    \frac{(Th)^2}{k},
\]
and for every $v$,
\[
    \sum_{a=1}^k g_{va}^2\le h^2.
\]
Thus
\[
    Q
    \ge
    h^2\left(\frac{T^2}{k}-T\right).
\]
Every pair counted by $Q$ contributes at least $8r$ to the optimal cost, so
\[
    \OPT
    \ge
    8rh^2\left(\frac{T^2}{k}-T\right).
\]
If $T\le k$, we are done.  Otherwise, using
$rh^2\ge \OPT/(32k)$,
\[
    \frac{T^2}{k}-T\le 4k.
\]
Hence $T^2-kT-4k^2\le0$, and therefore
\[
    T\le \frac{1+\sqrt{17}}2\,k<3k.
\]
This proves the lemma.
\end{proof}

\begin{lemma}[Total size]
\label{lem:total-size}
The set $S$ satisfies
\[
    |S|
    \le
    k\left(\frac{c+1}{\varepsilon}\right)^{O(\ddim)}\log(2n).
\]
\end{lemma}

\begin{proof}
By \cref{lem:high-points}, $|H|\le k$.  By \cref{lem:one-ring}, every scale
ring contributes at most
\[
    3k\left(\frac{C}{\tau}\right)^{\ddim}
\]
representatives.  There are
\[
    J+1\le 2+\log_2(kn^2)=O(\log n)
\]
rings.  Substituting $\tau=\varepsilon/(4(c+1))$ proves the claim.
\end{proof}

\subsubsection{Weighted Implementation and Composition}

We use the density-adaptive construction on weighted instances only when the
weights are integer multiplicities.  For a weighted support \((Y,w)\), let
\(Y^{(w)}\) be the multiset containing \(w(y)\) copies of each \(y\in Y\), and
let
\[
    N=\sum_{y\in Y}w(y).
\]
The quantities \(\Gamma_m(y)\), \(\gamma_m(y)\), and \(\rho(y)\) are those
obtained by applying the ordinary density-adaptive construction to
\(Y^{(w)}\): nearest-neighbor sums are taken over copies, with
\(2\le m\le N\), while the representative set is stored only on the support
\(Y\).

\begin{lemma}[Weighted density-adaptive construction]
\label{lem:weighted-density-adaptive}
Let \((Y,w)\) be a weighted metric instance with integer multiplicity weights
of total weight \(N\) and doubling dimension at most \(\ddim\).  Suppose that
\[
    \OPT_{Y,w}\le \widehat{\OPT}\le c\OPT_{Y,w}.
\]
Then the weighted density-adaptive construction gives an
\(\varepsilon\)-movement coreset of \((Y,w)\) of size
\[
    k\left(\frac{c+1}{\varepsilon}\right)^{O(\ddim)}\log(2N).
\]
Moreover, its movement rule satisfies the movement-charge bound of
\Cref{lem:density-movement-charge}, with \(Y^{(w)}\) in place of \(X\).
\end{lemma}

\begin{proof}
Apply the unweighted density-adaptive construction to the multiset
\(Y^{(w)}\).  Multiplicities do not change the doubling dimension after
coincident copies are identified. By
\Cref{clm:unsplitting-representatives}, an optimal partition of \(Y^{(w)}\)
can be chosen without splitting coincident copies. Hence its optimum value is
exactly \(\OPT_{Y,w}\). Therefore \Cref{lem:density-preservation},
\Cref{lem:density-movement-charge}, and \Cref{lem:total-size}, applied with
input size \(N\), give the preservation, movement-charge, and size guarantees
for \(Y^{(w)}\).

It remains only to write the result in weighted form.  Since coincident copies
have distance zero, the nets in the construction may be chosen with at most one
representative at each support point.  We store these support representatives
and give each representative the total multiplicity of the copies mapped to it.
With this compressed notation, the moved multiset cost is exactly the weighted
cost.  Thus the same guarantees hold for the weighted instance \((Y,w)\).
\end{proof}

\begin{lemma}[Computing weighted densities]
\label{lem:weighted-density-runtime}
Let \(M=|Y|\).  All values \(\rho(y)\) for the weighted instance \((Y,w)\) can
be computed from the support in time
\[
    O\!\left(M^2(\log M+\log N)\right).
\]
\end{lemma}

\begin{proof}
Fix \(y\in Y\), and sort the support points as \(y_1,\ldots,y_M\) by increasing
distance from \(y\).  Let
\[
    P_q=\sum_{i\le q}w(y_i),
    \qquad
    D_q=\sum_{i\le q}w(y_i)\dist(y,y_i).
\]
If \(P_{q-1}<m\le P_q\), then the \(m\) nearest copies to \(y\) consist of all
copies of \(y_1,\ldots,y_{q-1}\), together with \(m-P_{q-1}\) copies of \(y_q\).
Therefore
\[
    \gamma_m(y)
    =
    \frac{D_{q-1}+(m-P_{q-1})\dist(y,y_q)}{m}.
\]
Because the support points are sorted by distance from \(y\),
\(D_{q-1}\le P_{q-1}\dist(y,y_q)\), and the displayed expression is
nondecreasing in \(m\).  The term \(\widehat{\OPT}/(km^2)\) is decreasing in
\(m\).  Hence the minimum of their maximum over the interval is attained at an
endpoint or at an integer nearest their crossing point, which can be found by
binary search in \(O(\log N)\) time.
There are \(M\) intervals for this \(y\).  Sorting for all \(y\)'s and scanning
all intervals gives the claimed running time.
\end{proof}

\begin{corollary}[Density-adaptive coreset after annular reduction]
\label{cor:density-after-annular}
Let \((Y,w)\) be the \((\varepsilon/10)\)-movement coreset obtained from the
annular construction applied to \(X\), and suppose that
\[
    \OPT_{Y,w}\le \widehat{\OPT}\le c\OPT_{Y,w}
\]
for a universal constant \(c\).  Run the weighted density-adaptive construction
on \((Y,w)\) with estimate \(\widehat{\OPT}\) and accuracy \(\varepsilon/10\).
If  \(T\) is the resulting representative set, and the final weights are
recounted according to the nearest-neighbor rule \(\nu:X\to T\), then the
resulting weighted set is an \(\varepsilon\)-movement coreset of \(X\) under
\(\nu\), and has size
\[
    k\left(\frac{1}{\varepsilon}\right)^{O(\ddim)}\log n.
\]
If \(M=|Y|\), the additional time after \((Y,w)\) is available is
\[
    O\!\left(M^2(\log M+\log n)\right).
\]
\end{corollary}

\begin{proof}
Let \(\pi:X\to Y\) be the annular movement rule.  The annular construction
has the movement-charge bound of \Cref{lem:annular-movement-charge}.  Applying
\Cref{lem:weighted-density-adaptive} to \((Y,w)\) gives a weighted set
\((T,u')\) with movement rule \(\sigma:Y\to T\), and gives the corresponding
movement-charge bound for \(\sigma\).  Let \(\phi=\sigma\circ\pi:X\to T\).
By \Cref{lem:bounded-movement-composition}, the movement-charge coefficient of
\(\phi\) is at most
\[
    \frac{\varepsilon}{10}
    +
    \left(1+\frac{\varepsilon}{10}\right)\frac{\varepsilon}{10}
    <
    \varepsilon.
\]
Thus \(\phi\) satisfies the hypothesis of
\Cref{lem:nearest-neighbor-reweighting}.

Since \(T\subseteq Y\subseteq X\), let \(\nu:X\to T\) be the nearest-neighbor
map with fixed tie-breaking, and set \(u(t)=|\nu^{-1}(t)|\). Since
\(\phi(x)\in T\), for every \(x\in X\),
\[
    \dist(x,\nu(x))
    \le
    \dist(x,\phi(x)).
\]
Hence, by \Cref{lem:nearest-neighbor-reweighting}, \((T,u)\) is an
\(\varepsilon\)-movement coreset of \(X\) under \(\nu\).

The size bound follows from \Cref{lem:weighted-density-adaptive}, because the
total weight of \((Y,w)\) is \(n\).  Finally,
\Cref{lem:weighted-density-runtime} gives the displayed additional running
time.
\end{proof}

\subsubsection{Running Time}

The direct density-adaptive implementation computes all pairwise distances and
therefore takes \(O(n^2\log n)\) time.  The near-linear implementation in
\Cref{thm:density-adaptive} avoids this by first applying the annular
construction and then running the weighted density-adaptive construction only
on the resulting small weighted instance.

Let \(M=|Y|\) be the size of the annular coreset.  By
\Cref{thm:annular-coreset}, $M\le \left(\frac{k}{\varepsilon}\right)^{O(\ddim)}\log n .$

\Cref{lem:weighted-density-runtime} computes the weighted density values in $O\!\left(M^2(\log M+\log n)\right)$ time, and the net construction on the support takes no more, thus the total time after the annular
construction is $\left(\frac{k}{\varepsilon}\right)^{O(\ddim)}\operatorname{polylog} n .$

Together with the near-linear annular construction and the Czumaj--Sohler
polylogarithmic-factor approximation algorithm, this gives a randomized
near-linear-time implementation in the original input size \(n\).

\subsubsection{Proof of the Density-Adaptive Coreset Theorem}

\DensityAdaptiveCoreset*

\begin{proof}
The existential statement follows from the density-adaptive construction
analyzed above, using any value \(\widehat{\OPT}\) satisfying
\(\OPT\le\widehat{\OPT}\le O(\OPT)\).  \Cref{lem:density-preservation} and
\Cref{lem:total-size} give an \(\varepsilon\)-movement coreset of size
\[
    k\left(\frac{1}{\varepsilon}\right)^{O(\ddim)}\log n .
\]

For the randomized implementation, run the polylogarithmic-factor approximation
algorithm of Czumaj and Sohler~\cite{CzumajSohler2010}, and let \(A\) be twice
the cost of the partition it returns. With high probability,
\(2\OPT\le A\le 2L\OPT\), where \(L=(\log n)^{O(1)}\). Condition on this event.
First construct an \((\varepsilon/20)\)-movement annular coreset \((Y,w)\) of
\(X\) using \Cref{thm:annular-coreset}. Since \((Y,w)\) is an
\((\varepsilon/20)\)-movement coreset and \(\varepsilon<1/2\),
\[
    \OPT_{Y,w}\le A\le 4L\OPT_{Y,w}.
\]

For \(i=0,\ldots,\lceil\log_2(4L)\rceil\), set \(A_i=A/2^i\). On the same
weighted instance \((Y,w)\), run the weighted density-adaptive construction
with estimate \(A_i\) and accuracy \(\varepsilon/20\), aborting and discarding
this run if its number of representatives exceeds the bound of
\Cref{lem:weighted-density-adaptive} for a universal constant \(c\). Let \(T\) be the representative set of the last run that meets
the size bound, with index \(j\). Some \(i^\star\) satisfies
\(\OPT_{Y,w}\le A_{i^\star}<2\OPT_{Y,w}\).
By \Cref{lem:weighted-density-adaptive}, this run meets
the size bound and is therefore not discarded.
Hence \(j\ge i^\star\) and \(A_j<2\OPT_{Y,w}\).
The proof of \Cref{lem:density-movement-charge} uses only
this upper bound, so its bound also holds for the selected run.
The lower bound on the estimate is needed only for
\Cref{lem:total-size}; here, the size bound holds
by the choice of \(T\).
Applying \Cref{lem:bounded-movement-composition} and \Cref{lem:nearest-neighbor-reweighting} gives, after reweighting, an \(\varepsilon\)-movement coreset of \(X\) of the claimed size.   The annular construction is near-linear
in \(n\), and the weighted density computations are performed only on the
annular coreset, whose size is \((k/\varepsilon)^{O(\ddim)}\log n\).  By
\Cref{lem:weighted-density-runtime}, the \(O(\log L)\) weighted calls add only
\((k/\varepsilon)^{O(\ddim)}\operatorname{polylog} n\) time.  Thus the
implementation is near-linear in the original input size \(n\).
\end{proof}

\subsection{A Lower Bound for Movement Coresets}

We next show that the logarithmic dependence on \(n\) in the above movement-coreset bounds is unavoidable, already on the line and for \(k=2\). This logarithmic dependence is perhaps surprising in light of the fact that analogous coresets for \(k\)-median in Euclidean metrics can have size independent of \(n\) \cite{BravermanFLSY17}.

\begin{proposition}[A logarithmic lower bound for movement coresets]
\label{prop:movement-logn-necessary}
There are one-dimensional Euclidean instances of constant doubling dimension for which every \(1/100\)-movement coreset for min-sum \(2\)-clustering has size \(\Omega(\log n)\), where \(n\) is the input size.
\end{proposition}

\begin{proof}
Let \(M\) be large, let
\[
    L=\left\lfloor \frac12\log_2 M\right\rfloor,
    \qquad
    r_i=2^{-i},
\]
and consider the multiset
\[
    X=\{0_d\}^M\cup\{x_i=r_i e_1:\ i\in[L]\}.
\]
This instance lies on a line, hence has constant doubling dimension, and its
size is \(M+L=\Theta(M)\).  We prove that no movement coreset \((S,w)\) with
\(|S|<L\) can preserve all \(2\)-partitions within a small constant factor.
Fix any movement rule \(\pi\) for such a coreset.

First, all copies of \(0_d\) must be mapped to the same representative.
Otherwise, the partition separating the zero copies from the nonzero points has
original zero cost on the zero cluster, but after applying \(\pi\) that cluster
has positive cost.  The cost of the other cluster is at most \(L^2\), whereas
the represented zero-cluster cost is at least \(2(M-1)2^{-L}\), so for large
\(M\) this already violates the movement guarantee.  Let \(z\in S\) be the
common representative of the zero copies.

We now look at the distance from each representative \(\pi(x_i)\) to \(z\).  We
claim that for some \(i\), this distance differs from the original distance
\(\dist(x_i,0_d)=r_i\) by more than \(r_i/4\).  Indeed, suppose otherwise.
Then for every \(i\),
\[
    \dist(\pi(x_i),z)\in\left[\frac34 r_i,\frac54 r_i\right].
\]
These intervals are disjoint: for consecutive radii,
\(\frac54 r_{i+1}=\frac58 r_i<\frac34 r_i\).  Hence the \(L\) numbers
\(\dist(\pi(x_i),z)\) are all distinct, and therefore the representatives
\(\pi(x_i)\) are \(L\) distinct points of \(S\).  This contradicts \(|S|<L\).
Thus there is some \(i\in[L]\) such that
\[
    \left|\dist(\pi(x_i),z)-r_i\right|>\frac14 r_i .
\]

Consider the \(2\)-partition \(\calP_i=\{A_i,B_i\}\), where
\[
    A_i=\{0_d\}^M\cup\{x_i\},
    \qquad
    B_i=X\setminus A_i .
\]
The contribution of \(A_i\) to the original cost is \(2Mr_i\), while its
contribution after applying \(\pi\) is \(2M\dist(\pi(x_i),z)\).  Therefore the
error from this part alone is larger than \(Mr_i/2\).  Both the original and the
projected costs of \(B_i\) are at most \(L^2\), since \(B_i\) has at most \(L\)
points and diameter at most \(1\).  Hence
\[
    \left|
    \cost(\calP_i)
    -
    \cost(\pi(\calP_i))
    \right|
    \ge
    \frac12 Mr_i-2L^2.
\]
As \(i\le L\), we have \(Mr_i\ge M2^{-L}\ge\sqrt M\), and thus for large enough
\(M\) the right-hand side is at least \(Mr_i/4\).  On the other hand,
\[
    \cost(\calP_i)\le 2Mr_i+L^2\le 3Mr_i .
\]
The relative error is therefore at least \(1/12\), contradicting the
\(1/100\)-movement guarantee.  Since \(L=\Theta(\log |X|)\), the claimed
lower bound follows.
\end{proof}

\subsection{Streaming Applications}

The coreset constructions also imply two streaming applications.  For insertion-only
streams, we combine the two-pass streaming algorithm of Czumaj and
Sohler~\cite{CzumajSohler2010} with the annular construction.  The annular
construction also gives a multipass dynamic streaming algorithm by repeatedly
improving a coarse reference solution and then building the final coreset.

\ThreePassAnnularStreaming*

\begin{proof}
Run the two-pass streaming algorithm of Czumaj and
Sohler~\cite{CzumajSohler2010}.  It returns a constant-factor solution for
balanced \(k\)-median, and hence, by the equivalence between balanced
\(k\)-median and min-sum used above, a constant-factor reference solution for
min-sum.  Let \(c_1,\ldots,c_k\) be the corresponding centers.  The cost of
this reference solution is at most \(c\OPT\), for a universal constant \(c>0\).

In a third pass, run the annular construction with this reference solution,
using the cost of the reference solution as the reference cost, and using the
same radius sequence, terminal level, and accuracy parameter, with \(c\) as the
approximation constant in \Cref{lem:annular-preservation}.  The
representatives and weights are maintained online; the induced online movement
rule is used only in the analysis of the value estimate.  For every center
\(c_i\) and annular level
\(j\), maintain the corresponding
\(\eta r_j\)-packing \(N_{i,j}\subseteq X\cap B(c_i,r_j)\), together with
weights on the representatives.  When a point \(x\) arrives, consider the
centers \(c_i\) whose outer ball contains \(x\), and let \(\ell_i(x)\) be the
annular level of \(x\) around such a center.  Scan the nets
\(N_{i,\ell_i(x)}\).  If for some scanned net \(x\) is not covered at its scale,
add \(x\) to the first such net in a fixed order, initialize its weight to
\(1\).  If all scanned nets cover \(x\), increase the weight of a nearest
representative \(s\) in their union.  Let \(S\) be the union of all maintained
nets, and denote these online weights by \(w'\).

For the analysis, let \(\pi':X\to S\) be the implicit movement rule induced by
these updates: a point that creates a representative is mapped to that
representative, and a point whose arrival increases the weight of \(s\) is
mapped to \(s\).  This map is not stored by the algorithm.  By construction,
\[
    \forall s\in S\qquad w'(s)=|(\pi')^{-1}(s)|.
\]

Under this implicit rule, a point that creates a representative incurs zero
movement, and every other point is sent to a representative that covers it at
the annular level of a closest reference center.  Hence the movement bound of
\Cref{lem:annular-movement} holds for \(\pi'\), and the proof of
\Cref{lem:annular-preservation} applies unchanged.  Thus \((S,w')\), with
movement rule \(\pi'\), is an \(\varepsilon\)-movement coreset.

The algorithm then computes the optimal weighted min-sum value on \((S,w')\)
and outputs this value as \(\widehat{\OPT}\).  We now show that this is a
\((1\pm\varepsilon)\)-estimate of the original optimum.  By
\Cref{cor:weighted-existence}, there is a partition of \(S\) whose weighted
cost is at most \((1+\varepsilon)\OPT\), and therefore
\[
    \widehat{\OPT}\le (1+\varepsilon)\OPT .
\]
Conversely, let \(\calQ\) be an optimal partition of \((S,w')\), and lift it
to a partition \(\calP\) of \(X\) using \(\pi'\).  Then
\(\cost_{S,w'}(\calQ)=\cost(\pi'(\calP))\), and the movement-coreset guarantee
gives
\[
    \widehat{\OPT}
    =
    \cost_{S,w'}(\calQ)
    \ge
    (1-\varepsilon)\cost(\calP)
    \ge
    (1-\varepsilon)\OPT .
\]

With one additional pass over the stream, keep the same representative set
\(S\) and recount the weights according to the final nearest-neighbor movement
rule \(\pi:X\to S\), namely set \(w(s)=|\pi^{-1}(s)|\).  Since \(\pi(x)\) is a
nearest representative in the final set \(S\), while \(\pi'(x)\in S\), we have
\(\dist(x,\pi(x))\le\dist(x,\pi'(x))\) for every \(x\in X\).  Therefore the
annular movement bound also holds for \(\pi\), and by \Cref{lem:nearest-neighbor-reweighting}, the bound in  \Cref{lem:annular-preservation} applies to the recounted weighted set
\((S,w)\).  This proves the stated coreset guarantee after the additional pass.

Since each \(N_{i,j}\) is still an \(\eta r_j\)-packing in \(B(c_i,r_j)\),
\Cref{lem:annular-size} gives
\[
    |S|
    \le
    \left(\frac{k}{\varepsilon}\right)^{O(\ddim)}\log n .
\]
Adding the \(\operatorname{poly}(k,\log n)\) space used by the Czumaj--Sohler streaming
algorithm gives the claimed space bound.  The guarantees above are conditioned
on the success event of the Czumaj--Sohler primitive, which holds with high
probability; conditional on this event, the remaining steps are deterministic.
\end{proof}
\paragraph{Grid version of the annular construction.}
For the dynamic and Euclidean setting we use the same annular construction, but replace the nets by fixed grids.  This makes the set of possible representatives independent
of the active point set, which is necessary when points may later be deleted.

Let
\[
    \calR=\{R_1,\ldots,R_k\}
\]
be a reference partition with centers \(c_1,\ldots,c_k\), and let
\(\widehat{\OPT}=\cost(\calR)\).  As before, set
\[
    r_0=20\widehat{\OPT},
    \qquad
    r_j=\frac{r_0}{2^j},
    \qquad
    j=0,\ldots,L,
\]
where \(L=\lceil \log_2(40n^{10})\rceil\).  For each \(i\in[k]\) and level
\(j\), take a grid around \(c_i\) with side length
\[
    \ell_j=\frac{\eta r_j}{\sqrt d}.
\]
Let \(S\) contain all grid points whose cells intersect \(B(c_i,r_j)\), over all
\(i\) and \(j\).  Let \(\pi:X\to S\) be the nearest-grid movement rule,
which maps each point \(x\in X\) to the closest grid point \(\pi(x)\in S\).

The same movement bound as in the annular construction holds.
Consequently, all pairwise movement estimates from the annular proof remain
unchanged.  The only difference is the size bound.  At each level and for each
center, the number of grid points is at most
\[
    \left(\frac{Ck\sqrt d}{\varepsilon}\right)^d
\]
for a universal constant \(C\).  Therefore,
\[
    |S|
    \le
    k\left(\frac{Ck\sqrt d}{\varepsilon}\right)^d\log n .
\]
The advantage of this grid version is that insertions and deletions only change
the weights of fixed representatives: an insertion increases the weight of
\(\pi(x)\), and a deletion decreases the weight of the same representative.
\MultiPassDynamicStreaming*
\begin{proof}
We first apply the dynamic \(k\)-median coreset construction
of~\cite{BravermanFLSY17} to the original stream on \([\Delta]^d\), with
constant accuracy. We compute a constant-factor \(k\)-median solution
\(C=\{c_1,\ldots,c_k\}\) from the resulting coreset. If \(D\) is its $k$-median cost
on \(X\), then \(D\le O(1)\OPT_{\mathrm{med}}\le O(1)\OPT\).

We now choose the Gaussian map \(f\) with target dimension
\[
    t=O\left(
            \varepsilon^{-2}\bigpar{
            \ddim\log(1/\varepsilon)+\log k+\log\log n}
    \right).
\]
For each \(x\in X\), let \(\sigma(x)\in C\) be the closest center to \(x\) in the
original metric.  Since \(C\) and \(\sigma\) are fixed before \(f\) is sampled,
a standard fixed-solution JL estimate gives, with high probability,
\[
    \sum_{x\in X}\|f(x)-f(\sigma(x))\|
    \le
    O(1)\sum_{x\in X}\|x-\sigma(x)\|
    \le
    O(nD).
\]
Thus the projected centers \(f(C)\) have \(k\)-median cost \(O(nD)\) for the
embedded instance.  If \(\calP_0\) is the partition of \(f(X)\) induced by the
nearest center in \(f(C)\), then by the triangle inequality,
\[
    \cost(\calP_0)
    \le
    2n\sum_{x\in X}\dist(f(x),f(C))
    \le
    O(n^2D)
    \le
    O(n^2)\OPT .
\]
By \Cref{thm:main-jl-ddim}, the optimum min-sum cost in the embedded instance is
within a constant factor of \(\OPT\).  Hence \(f(C)\) gives a sufficiently
coarse reference solution for the embedded instance.

We now run the grid-based annular construction on the embedded instance in
\(\mathbb R^t\).  The embedded points are not necessarily contained in a finite
integer grid, but this causes no problem: the dynamic \(k\)-median algorithm was
used only before the projection, on the original universe \([\Delta]^d\).  Once
the reference centers are fixed, the grid construction defines a finite set of
representatives in \(\mathbb R^t\), and subsequent passes over the original
stream only update the weights of these fixed representatives.

The movement inequalities imply an improvement recurrence of the form
\[
    \cost(\calP_{r+1})
    \le
    O(1)\OPT+\frac12\cost(\calP_r).
\]
After \(O(\log n)\) passes, this gives a constant-factor approximate reference
partition.

In one additional pass, run the same grid construction with movement parameter
\[
    \eta=\Theta\left(\frac{\varepsilon}{k}\right).
\]
The movement guarantee then gives, for every partition \(\calP\) of the final
active set \(X\) into at most \(k\) parts,
\[
    \left|
    \cost(\calP)
    -
    \cost(\pi(\calP))
    \right|
    \le
    \varepsilon\cost(\calP).
\]
Since the grid construction is now applied in dimension \(t\), it uses at most
\[
    \left(\frac{Ck\sqrt t}{\varepsilon}\right)^t
\]
representatives per center and per level.  There are \(k\) centers and
\(O(\log n)\) levels, giving total space
\[
    k\left(\frac{Ck\sqrt t}{\varepsilon}\right)^t\log n .
\]
After adjusting the constant hidden in the definition of \(t\), this is bounded
by
\[
    \left(\frac{k t}{\varepsilon}\right)^t\log n .
\]
Together with the initial constant-factor \(k\)-median computation, the total
number of passes is \(O(\log n)\), and the space is
\[
    \operatorname{poly}(d,k,\log(n))+\left(\frac{k t}{\varepsilon}\right)^t\log n .
\]
As in the insertion-only case, solving the final weighted instance gives a value
\(\widehat{\OPT}\in(1\pm\varepsilon)\OPT\).  Lifting an optimal weighted
partition back to \(X\) through \(\pi\) gives a \(k\)-partition
\(\widehat{\calP}\) with
\(\cost(\widehat{\calP})\le(1+O(\varepsilon))\OPT\), and rescaling
\(\varepsilon\) gives the stated bound.
\end{proof}

\section{Coreset Properties}
\label{sec:coreset-attributes}

Movement coresets compress the input while approximately preserving the cost of every
partition.  This raises several
questions beyond their construction:
whether optimization can be performed directly on the weighted instance,
whether the same constructions extend to related objectives, and whether
the guarantees remain valid when coresets or movement rules are combined. The results in this section demonstrate that movement-coresets
support substantially more than partition-cost preservation.

Our results in this section rely on two different coreset guarantees.
Some use only the fact that the cost of every partition is preserved,
whereas others require explicit control of the movement (distances between input points
and their representatives). Since a point \(x\) in a part \(A\) participates in
\(|A|-1\) pairwise distances, the effect of moving it is
naturally controlled by the quantity
\(|A|\sum_{x\in A}\dist(x,\pi(x))\).  This motivates the notion of a
\(\epsilon\)-bounded-movement coreset, defined below.  Our coreset constructions
satisfy this bounded-movement guarantee.  The triangle inequality converts this guarantee into the usual movement-coreset guarantee, while the explicit control of the
movements is useful for several additional applications.  We then show that this stronger guarantee remains useful across related
clustering settings and under natural coreset operations. 
We show that optimization does not require splitting representative weights
(\Cref{clm:unsplitting-representatives} and \Cref{cor:weighted-existence}), extend the
constructions to related objectives
(\Cref{lem:powered-density-coreset} and \Cref{prop:shifted-centers-balanced-kmeans}), and
determine when coresets can be combined or reweighted without losing their
guarantees
(\Cref{lem:color-balanced-mergeability}, \Cref{lem:movement-composition} and \Cref{lem:nearest-neighbor-reweighting}).
\begin{definition}[Bounded-movement coreset]
\label{def:bounded-movement-coreset}
A weighted set \((S,w)\), with \(S\subseteq X\), is a
\emph{\(\epsilon\)-bounded-movement coreset} if there is a map \(\pi:X\to S\),
consistent with the weights, such that every \(k\)-partition \(\calP\) of
\(X\) satisfies
\begin{gather}
    \sum_{A\in\calP}|A|\sum_{x\in A}\dist(x,\pi(x))
    \le
    \epsilon\cdot\cost(\calP). \label{eq:movement_cost_bound}
\end{gather}
It is sometimes useful to restrict this requirement to a specified set of feasible
partitions.
\end{definition}

This bound controls, for every partition, the distances between the input
points and their representatives relative to the cost of the partition.  By the
triangle inequality, every \(\varepsilon\)-bounded-movement coreset is a
\(2\varepsilon\)-movement coreset.

\subsection{Unsplitting and Optimization}

A \(k\)-partition of the multiset represented by \((S,w)\) may assign different
copies of the same representative to different clusters. We call a partition
\emph{consistent} if all copies of each \(s\in S\) belong to the same part.
Such a partition can be identified with a partition of \(S\).
The next claim shows that consistency can be imposed without increasing
the min-sum cost.

\begin{claim}[Unsplitting representatives]
\label{clm:unsplitting-representatives}
Let \((S,w)\) be a weighted set with integer weights. Suppose we partition
the multiset represented by \((S,w)\) into at most \(k\) parts, possibly
splitting copies of the same \(s\in S\) between different parts. Then there
is a consistent partition into at most \(k\) parts, of no larger cost.
\end{claim}

\begin{proof}
If the copies of some \(s\in S\) are split between several parts, move all of
them to the part that minimizes their contribution to the cost. Since copies of
\(s\) are at distance zero from one another, this minimization is based on the
interaction of \(s\) with (copies of) other points in $S$.
Repeating this for every split representative gives the claim.
\end{proof}

Consequently, optimizing over consistent partitions of \((S,w)\) yields
a partition of \(X\) whose cost is comparable to the optimal min-sum \(k\)-clustering cost of \(X\).

\begin{corollary}[Optimization on a movement coreset]
\label{cor:weighted-existence}
Let \((S,w)\) be an \(\varepsilon\)-movement coreset with movement rule
\(\pi:X\to S\). There exists a partition \(\calQ^*\) of \(S\) into at most
\(k\) parts such that
\[
    \cost_{S,w}(\calQ^*)\le (1+\varepsilon)\OPT .
\]
\end{corollary}

\begin{proof}
Let \(\calP^*\) be an optimal partition of \(X\). The partition \(\pi(\calP^*)\) assigns the copies of each representative in \(S\) to at
most \(k\) parts, and  by the movement guarantee the cost of \(\pi(\calP^*)\) is at most
\((1+\varepsilon)\OPT\). By \Cref{clm:unsplitting-representatives}, we may
consolidate the copies of every representative without increasing the cost,
which gives the desired weighted partition of \(S\).
\end{proof}

\subsection{Costs with Powers of Distances}

Fix a constant \(z\ge1\).  For a set \(A\subseteq X\) and a partition \(\calP\) of $X$, define
\[
    \cost_z(A)=\sum_{x,y\in A}\dist(x,y)^z,
    \qquad
    \cost_z(\calP)=\sum_{A\in\calP}\cost_z(A),
     \qquad 
     \OPT_z=\min_{\calP}\cost_z(\calP).
\]
In this subsection, we use the power-$z$ analogue of
\Cref{def:bounded-movement-coreset} which is obtained by raising each movement
distance to power $z$ and replacing \(\cost\) by \(\cost_z\).

\begin{lemma}[Density-adaptive coresets for distance powers]
\label{lem:powered-density-coreset}
For every \(z\ge1\) and \(0<\varepsilon<1/2\), an adaptation of our density-adaptive construction produces a coreset of size
\[
    k\left(\frac{z}{\varepsilon}\right)^{O(\ddim)}\log n
\]
whose movement rule $\pi$ satisfies, for every \(k\)-partition \(\calP\),
\[
    \sum_{A\in\calP}|A|
    \sum_{x\in A}\dist(x,\pi(x))^z
    \le
    \left(\frac{\varepsilon}{16z}\right)^z\cost_z(\calP).
\]
Thus, the resulting coreset satisfies the power-\(z\) analogue of
\Cref{def:bounded-movement-coreset} with
\(\epsilon^\prime=(\varepsilon/(16z))^z\). Consequently, it preserves \(\cost_z\) within
a factor \(1\pm\varepsilon\) for every \(k\)-partition.
\end{lemma}

\begin{proof}
In the density-adaptive construction, replace the movement scale by
\[
    r(x)
    =
    \min_{2\le m\le n}
    \max\left\{
        \left(
            \frac1m\sum_{w\in\Gamma_m(x)}\dist(x,w)^z
        \right)^{1/z},
        \left(\frac{\widehat{\OPT}_z}{km^2}\right)^{1/z}
    \right\},
\]
where \(\OPT_z\le\)\(\widehat{\OPT}_z\le C\cdot\OPT_z\) for a constant \(C\ge1\).  Construct the same rings and nets at distance scale \(\eta r(x)\), where
\(\eta=\varepsilon/(32Cz)\). As before, let \(S_r\) denote the representatives chosen for the ring \(\{x\in X:r\le r(x)<2r\}\), and let \(S\) be the union of these representative sets.  If \(|A|=m\), the definition of \(r(x)\) gives
\[
    |A|\sum_{x\in A}\dist(x,\pi(x))^z
    \le
    \eta^z\cost_z(A)
    +
    \frac{\eta^z\widehat{\OPT}_z}{k}.
\]
Summing over the at most \(k\) parts of \(\calP\), and using
\(\widehat{\OPT}_z\le C\cdot\OPT_z\le C\cdot\cost_z(\calP)\), gives
\[
    \sum_{A\in\calP}|A|
    \sum_{x\in A}\dist(x,\pi(x))^z
    \le (1+C)\eta^z\cost_z(\calP)
    \le \left(\frac{\varepsilon}{16z}\right)^z\cost_z(\calP).
\]
This proves the displayed movement bound.

For all \(a,b\ge0\), by convexity we have
\begin{equation}
    \label{eq:convexity}
    (a+b)^z
    \le
    \left(1+\frac{\varepsilon}{2}\right)a^z
    +
    \left(1+\frac{4z}{\varepsilon}\right)^{z-1}b^z. 
\end{equation}
For every part \(A\in\calP\) and every pair points \(x,y\in A\), let
\[
    b_{x,y}
    =
    \dist(x,\pi(x))+\dist(y,\pi(y)).
\]
By the triangle inequality
\[
    |\dist(\pi(x),\pi(y))- \dist(x ,y)|
    \le
    b_{x,y}.
\]
Applying \eqref{eq:convexity} gives

\begin{gather}
    \dist(\pi(x),\pi(y))^z
     \le
    \left(1+\frac{\varepsilon}{2}\right)\dist(x,y)^z
    +
    \left(1+\frac{4z}{\varepsilon}\right)^{z-1}b_{x,y}^z ,\label{eq:first_powered_upperbound} 
    \\
    \dist(x,y)^z
     \le
    \left(1+\frac{\varepsilon}{2}\right)\dist(\pi(x),\pi(y))^z
    +
    \left(1+\frac{4z}{\varepsilon}\right)^{z-1}b_{x,y}^z. \label{eq:second_powered_upperbound}
\end{gather}
We next bound the second term in these inequalities. Since
\[
    b_{x,y}^z
    \le
    2^{z-1}\left(
        \dist(x,\pi(x))^z+\dist(y,\pi(y))^z
    \right),
\]
summing over all pairs of points in all parts and rearranging gives
\[
\begin{aligned}
    \sum_{A\in\calP}
    \sum_{\substack{x,y\in A}}b_{x,y}^z
    &\le
    2^z\sum_{A\in\calP}|A|
    \sum_{x\in A}\dist(x,\pi(x))^z,
\end{aligned}
\]
and since $1+4z/\varepsilon\leq5z/\varepsilon$, we get
\[
\begin{aligned}
    \left(1+\frac{4z}{\varepsilon}\right)^{z-1}
    \sum_{A\in\calP}
    \sum_{\substack{x,y\in A}}b_{x,y}^z
    \le
    2^z\left(\frac{5z}{\varepsilon}\right)^{z-1}
    \left(\frac{\varepsilon}{16z}\right)^z
    \cost_z(\calP)
    \le
    \frac{\varepsilon}{8}\cost_z(\calP),
\end{aligned}
\]

Summing \eqref{eq:first_powered_upperbound}  over all pairs gives
\[
    \cost_z(\pi(\calP))
    \le
    \left(1+\frac{\varepsilon}{2}
    +\frac{\varepsilon}{8}\right)\cost_z(\calP)
    \le
    (1+\varepsilon)\cost_z(\calP).
\]
Summing \eqref{eq:second_powered_upperbound} over all pairs gives
\[
    \cost_z(\pi(\calP))
    \ge
    \frac{1-\varepsilon/8}{1+\varepsilon/2}\cost_z(\calP)
    \ge
    (1-\varepsilon)\cost_z(\calP).
\]

It remains to bound the size of $S$.  Fix a radius \(r\), and let \(N_r\)
be a \((16r)\)-net of $S_r$.  The proof of
\Cref{lem:one-ring} applies with power-$z$ distances: for every \(v\in N_r\), more
than half of the \(m_v\) nearest points lie in \(B(v,4r)\), and
\(m_v>\sqrt{\widehat{\OPT}_z/(2^zkr^z)}\).  Taking
\(h=\lceil\sqrt{\widehat{\OPT}_z/(2^{z+2}kr^z)}\rceil\), The same pair count across distinct groups \(G_v\), now with every distance
raised to the \(z\)-th power, gives \(|N_r|<3k\). Each region contains at most
\((Cz/\varepsilon)^{O(\ddim)}\) representatives by packing. As in
\Cref{lem:high-points}, the set \(H\) has size at most \(k\), and there are
\(O(\log n)\) scales. This proves the stated size of the coreset.
\end{proof}

\begin{lemma}[Annular coresets for power-$z$ distances]
\label{lem:powered-annular-coreset}
For every \(z\ge1\) and \(0<\varepsilon<1/2\), an adaption of our annular
construction has size
\[
    \left(\frac{kz}{\varepsilon}\right)^{O(\ddim)}\log n
\]
and satisfies, for every \(k\)-partition \(\calP\),
\[
    \sum_{A\in\calP}|A|
    \sum_{x\in A}\dist(x,\pi(x))^z
    \le \left(\frac{\varepsilon}{16z}\right)^z\cost_z(\calP).
\]
Thus it satisfies the power-$z$ analogue of
\Cref{def:bounded-movement-coreset} with
\(\varepsilon^\prime=(\varepsilon/(16z))^z\).  Consequently, it preserves \(\cost_z\) within factor \(1\pm\varepsilon\) for every \(k\)-partition.
\end{lemma}

\begin{proof}
Run the annular construction with movement parameter
\(\eta=\varepsilon/(16C_0zk)\), where \(C_0>0\) is a sufficiently large absolute
constant.  The charging argument of \Cref{lem:annular-movement-charge}, with
every distance raised to the \(z\)-th power, bounds the power-$z$ movement cost by
\[
    \sum_{A\in\calP}|A|
    \sum_{x\in A}\dist(x,\pi(x))^z
    \le (C_0k\eta)^z\cost_z(\calP)
    = \left(\frac{\varepsilon}{16z}\right)^z\cost_z(\calP).
\]
As shown in the proof of \Cref{lem:powered-density-coreset}, this movement bound is sufficient to preserve \(\cost_z\) within factor
\(1\pm\varepsilon\).
\end{proof}

\subsection{Balanced Euclidean \texorpdfstring{\(k\)}{k}-Means}
The case \(z=2\) of the preceding subsection yields a guarantee also for balanced Euclidean \(k\)-means which holds for every possible choice of centers. To formulate this
extension, we need to define how a movement rule acts on the centers associated with a partition.

For a \(k\)-partition \(\calP\) of a Euclidean set $X$ and corresponding
centers \(\mathcal C=(c_A)_{A\in\calP}\), define the \emph{balanced \(k\)-means} cost by
\[
    \cost_{\mathrm{bkmeans}}(\calP, \mathcal{C})
    =
    \sum_{A\in\calP}|A|
    \sum_{x\in A}\dist(x,c_A)^2.
\]
For a set \(B\subseteq\R^d\), define its centroid by
\[
    \overline{x}_B
    :=
    \frac{1}{|B|}\sum_{x\in B}x.
\]

Let \(\pi:X\to S\) be a movement rule. For a partition \(\calP\) and
corresponding centers \(\mathcal C=(c_A)_{A\in\calP}\), define the movement of
the centers under \(\pi\) by
\[
    \pi(\mathcal C)
    :=
    \left(c_A+\overline{x}_{\pi(A)}-\overline{x}_A\right)_{A\in\calP}.
\]

We call \((S,w)\) an \emph{\(\varepsilon\)-movement coreset for balanced
\(k\)-means} if, for every \(k\)-partition \(\calP\) and corresponding
centers \(\mathcal C=(c_A)_{A\in\calP}\),
\[
    \cost_{\mathrm{bkmeans}}\bigl(\pi(\calP),\pi(\mathcal C)\bigr)
    \in
    (1\pm\varepsilon)
    \cost_{\mathrm{bkmeans}}(\calP,\mathcal C).
\]
\begin{proposition}[Preservation under shifted centers]
\label{prop:shifted-centers-balanced-kmeans}
Let \(X\subseteq\mathbb R^d\), let \(0<\varepsilon\le1\), and let \((S,w)\)
be an \(\varepsilon\)-movement coreset for \(\cost_2\), with movement rule
\(\pi:X\to S\). Then \((S,w)\) is also an \(\varepsilon\)-movement coreset for
balanced \(k\)-means.
\end{proposition}

\begin{proof}
For every \(A\in\calP\), the Euclidean variance identity gives
\[
    |A|\sum_{x\in A}\dist(x,c_A)^2
    =
    \frac12\sum_{x,y\in A}\dist(x,y)^2
    +
    |A|^2\dist(\overline{x}_A,c_A)^2,
\]
and similarly, for the moved multiset \(\pi(A)\),
\[
    |A|\sum_{x\in A}
    \dist\bigl(\pi(x),c_A+\overline{x}_{\pi(A)}-\overline{x}_A\bigr)^2
    =
    \frac12\sum_{x,y\in A}\dist(\pi(x),\pi(y))^2
    +
    |A|^2
    \dist\bigl(
        \overline{x}_{\pi(A)},
        c_A+\overline{x}_{\pi(A)}-\overline{x}_A
    \bigr)^2.
\]
Obviously,
\[
    \overline{x}_{\pi(A)}
    -
    \bigl(c_A+\overline{x}_{\pi(A)}-\overline{x}_A\bigr)
    =
    \overline{x}_A-c_A,
\]
and thus the last term in the two variance decompositions above has the same value. The proposition follows by summing over the parts and applying the movement-coreset guarantee for \(\cost_2\).
\end{proof}

\begin{corollary}[Density-adaptive coreset for balanced \(k\)-means]
\label{cor:balanced-kmeans-coreset-size}
For every Euclidean input \(X\), there is a density-adaptive \(\varepsilon\)-movement coreset for balanced
Euclidean \(k\)-means.  For every partition \(\calP\) and every corresponding
set of centers \(\mathcal C\), it preserves
\(\cost_{\mathrm{bkmeans}}(\calP,\mathcal C)\) within \(1\pm\varepsilon\)
after shifting the centers as in
\Cref{prop:shifted-centers-balanced-kmeans}, using $k\left(\frac{1}{\varepsilon}\right)^{O(\ddim)}\log n$ representatives.
\end{corollary}

\begin{proof}
The density-adaptive construction and analysis of
\Cref{lem:powered-density-coreset} give the displayed movement guarantee and
the stated size for balanced \(k\)-means.  Apply
\Cref{prop:shifted-centers-balanced-kmeans}.
\end{proof}

\subsection{Color-Balanced Mergeability}
A coreset construction is mergeable if coresets constructed independently
for two sets can be combined into a coreset for their union. In general,
movement coresets are not mergeable, because moving points from one set affects their distances to possibly many points from the other set. We show that
merging is possible when every cluster contains the same number of points
from the two sets. We call the points of the two sets red and blue and
refer to such partitions as color-balanced.

Let \(R\) and \(B\) be two disjoint sets of equal size, and let \(X\) be their union. We restrict attention to \(k\)-partitions of \(X\) in
which every part contains the same number of red and blue points.

\begin{lemma}[Color-balanced mergeability]
\label{lem:color-balanced-mergeability}
For each \(C\in\{R,B\}\), let \((S_C,w_C)\) be an
\(\varepsilon\)-bounded-movement coreset for \(C\), and let
\(\pi_C:C\to S_C\) be a the movement rule. Consider their weighted
union \((S,w)\), where  $w=w_R+w_B$.  Define \(\pi:X\to S\) by \(\pi|_C=\pi_C\) for each
\(C\in\{R,B\}\).  Then, for every \(k\)-partition \(\calP\) of \(X\)
satisfying for all $A\in \calP \quad |A\cap R|=|A\cap B|$
it holds that
\[
    \sum_{A\in\calP}|A|
    \sum_{x\in A}\dist(x,\pi(x))
    \le
    2\varepsilon\cost(\calP).
\]
\end{lemma}
\begin{proof}
Let \(\calP\) be a \(k\)-partition of \(X\) satisfying
\[
    \forall A\in \calP \qquad |A\cap R|=|A\cap B|
\]
For each \(A\in\calP\), write
\[
    A_R=A\cap R,
    \qquad
    A_B=A\cap B.
\]
Color balance gives \(|A|=2|A_R|=2|A_B|\). Let $\calP_R=\{A_R| A\in \calP\}$ and $\calP_B=\{A_B| A\in \calP\}$ be the induced partitions on each color. For every \(A\in\calP\),
\[
        \cost(A_R)+\cost(A_B)\leq \cost(A)
\]
Consequently,
\[
    \cost(\calP_R)+\cost(\calP_B)\le\cost(\calP).
\]

Thus
\[
\begin{aligned}
    \sum_{A\in\calP}|A|\sum_{x\in A}\dist(x,\pi(x))
    &=
    2\sum_{A\in\calP}
    \left(
        |A_R|\sum_{x\in A_R}\dist(x,\pi_R(x))
        +
        |A_B|\sum_{x\in A_B}\dist(x,\pi_B(x))
    \right)
    \\
    &\le
    2\varepsilon
    \bigl(\cost(\calP_R)+\cost(\calP_B)\bigr)
    \\
    &\le
    2\varepsilon\cost(\calP).
\end{aligned}
\]
\end{proof}

\subsection{Composition}

Movement coresets may be applied repeatedly: a coreset of a coreset is also a coreset of the original input, under the composed movement rule.  The errors
accumulate additively, as the next lemma makes precise.

\begin{lemma}[Composition of movement coresets]
\label{lem:movement-composition}
Let \((S,w)\) be an \(\varepsilon_1\)-movement coreset of \(X\), with movement
rule \(\pi_1:X\to S\).  Let \((T,u)\) be an \(\varepsilon_2\)-movement coreset of
the weighted instance \((S,w)\), interpreted as a multiset, with movement rule
\(\pi_2:S\to T\) and weights
\[
    u(t)=\sum_{s:\pi_2(s)=t} w(s).
\]
Then \((T,u)\) is an
\((\varepsilon_1+\varepsilon_2+\varepsilon_1\varepsilon_2)\)-movement coreset
of \(X\) under the composed movement rule \(\pi_2\circ\pi_1:X\to T\).
\end{lemma}

\begin{proof}
The weights are consistent with the composed rule because
\[
    u(t)
    =
    \sum_{s:\pi_2(s)=t}w(s)
    =
    \left|(\pi_2\circ\pi_1)^{-1}(t)\right|.
\]
Fix a \(k\)-partition \(\calP\) of \(X\), and let
\(\calQ=\pi_1(\calP)\) be the corresponding partition of the multiset
represented by \((S,w)\). The first coreset guarantee gives
\[
    \cost_{S,w}(\calQ)
    \in
    (1\pm\varepsilon_1)\cost(\calP).
\]
Since \(\calQ\) is a partition of this multiset, applying the second coreset
guarantee gives
\[
    \cost_{T,u}(\pi_2(\calQ))
    \in
    (1\pm\varepsilon_2)\cost_{S,w}(\calQ).
\]
Since \(\pi_2(\calQ)=(\pi_2\circ\pi_1)(\calP)\), the upper bound is at most
\[
    (1+\varepsilon_2)(1+\varepsilon_1)\cost(\calP)
    =
    (1+\varepsilon_1+\varepsilon_2+\varepsilon_1\varepsilon_2)\cost(\calP),
\]
and the lower bound is at least
\[
    (1-\varepsilon_2)(1-\varepsilon_1)\cost(\calP)
    \ge
    (1-\varepsilon_1-\varepsilon_2)\cost(\calP).
\]
This is stronger than the claimed symmetric error.
\end{proof}

The following lemma gives the corresponding composition result for
bounded-movement coresets.
\begin{lemma}[Composition of bounded-movement coresets]
\label{lem:bounded-movement-composition}
Let \((S,w)\) be both an \(\varepsilon_1\)-movement coreset and a
\(\delta_1\)-bounded-movement coreset of \(X\) under the movement rule
\(\pi_1:X\to S\). Let \((T,u)\) be a \(\delta_2\)-bounded-movement coreset of
the weighted instance \((S,w)\), interpreted as a multiset, under the movement
rule \(\pi_2:S\to T\), where
\[
    u(t)=\sum_{s:\pi_2(s)=t}w(s).
\]
Then \((T,u)\) is a $\bigl(\delta_1+(1+\varepsilon_1)\delta_2\bigr)$-bounded-movement coreset of \(X\) under the composed movement rule
\(\pi_2\circ\pi_1:X\to T\).
\end{lemma}
\begin{proof}
Fix a \(k\)-partition \(\calP\) of \(X\), and let
\(\calQ=\pi_1(\calP)\) be the corresponding partition of the multiset
represented by \((S,w)\).  By the triangle inequality,
\[
    \dist(x,\pi_2(\pi_1(x)))
    \le
    \dist(x,\pi_1(x))+\dist(\pi_1(x),\pi_2(\pi_1(x))).
\]
The first term is bounded by the assumed bound in $\eqref{eq:movement_cost_bound}$ for \(\pi_1\).  For
the second term, each cluster of \(\calQ\) has exactly the same number of copies
as the corresponding cluster of \(\calP\), so the assumed movement-charge bound
for \(\pi_2\) gives
\[
    \sum_{A\in\calP}|A|\sum_{x\in A}\dist(\pi_1(x),\pi_2(\pi_1(x)))
    \le
    \delta_2\cost_{S,w}(\calQ)\leq\delta_2(1+\varepsilon_1)\cost(\calP).
\]
Combining the two bounds proves the claim.
\end{proof}

\subsection{Movement Costs and Nearest-Neighbor Reweighting}

The constructions above sometimes analyze a convenient intermediate movement
rule before reweighting according to nearest neighbors. Since this reweighting
can only decrease the movement distance of each input point, the
bounded-movement guarantee is preserved, as formalized by the next lemma.

\begin{lemma}[Nearest-neighbor reweighting]
\label{lem:nearest-neighbor-reweighting}
Let \(S\subseteq X\), and let \((S,w_1)\) be an
\(\varepsilon\)-bounded-movement coreset under a movement rule
\(\pi_1:X\to S\). Let \(\pi_2:X\to S\) be the nearest-neighbor movement rule
with fixed tie-breaking, and let \(w_2\) be the corresponding weights. Then
\((S,w_2)\) is an \(\varepsilon\)-bounded-movement coreset under \(\pi_2\).
\end{lemma}
\begin{proof}
For every \(x\in X\), the nearest-neighbor property gives
\[
    \dist(x,\pi_2(x))
    \le
    \dist(x,\pi_1(x)).
\]
Therefore, for every \(k\)-partition \(\calP\) of \(X\),
\[
    \sum_{A\in\calP}|A|
    \sum_{x\in A}\dist(x,\pi_2(x))
    \le
    \sum_{A\in\calP}|A|
    \sum_{x\in A}\dist(x,\pi_1(x))
    \le
    \varepsilon\cost(\calP).
\]
Since the new weights are defined according to \(\pi_2\), they are consistent
with this movement rule. Thus the reweighted set is an
\(\varepsilon\)-bounded-movement coreset.
\end{proof}

\section{AI Disclosure}
The authors used ChatGPT (OpenAI) as a writing and research-development aid, 
both to improve the clarity of author-written text and to write up proof details from author-provided ideas, proof sketches, and mathematical guidance. All results, definitions, and proof strategies originate with the authors, who verified all AI-assisted content and are fully responsible for the integrity, accuracy, and originality of the entire paper.

{\small
\bibliographystyle{alphaurl}
\bibliography{references}
} 
 
\end{document}